\documentclass[a4paper,11pt]{article}

\usepackage[left=2.6cm, right=2.6cm, top=2.6cm, bottom=2.6cm]{geometry}

\usepackage[latin1]{inputenc}
\usepackage[T1]{fontenc}
\usepackage[english,ngerman]{babel}

\usepackage{amsthm}
\usepackage{amsmath}
\usepackage{amssymb}
\usepackage[lined,boxed,commentsnumbered]{algorithm2e}

\usepackage{paralist}
\usepackage[arrow, matrix, curve]{xy}

\usepackage{booktabs, array, dcolumn}
\usepackage{tabularx}

\usepackage{url}

\usepackage{float}
\usepackage{graphicx}
\usepackage{caption}
\usepackage{subcaption}
\usepackage{color}

\usepackage{dsfont}

\usepackage{verbatim}

\newtheoremstyle{j}%
{3pt}%
{3pt}%
{}%
{\parindent}%
{\bfseries}%
{.}%
{.5em}%
{}%

\theoremstyle{plain}

\newtheorem*{rem*}{Remark}
\newtheorem*{concl*}{Conclusion}
\newtheorem*{theorem*}{Theorem}
\newtheorem*{cor*}{Corollary}
\newtheorem*{algo*}{Algorithm}

\newtheorem{theorem}{Theorem}[section]
\newtheorem{lemma}[theorem]{Lemma}

\newtheorem{cor}[theorem]{Corollary}
\newtheorem{prop}[theorem]{Proposition}
\newtheorem{deff}[theorem]{Definition}
\newtheorem{rem}[theorem]{Remark}

\theoremstyle{definition}

\newtheorem*{example*}{Example}
\newtheorem{example}[theorem]{Example}
\newtheorem*{prob*}{Problem}

\newcommand{\Hil}{\mathcal{H}}

\newcommand{\N}{\mathbb{N}}

\newcommand{\R}{\mathbb{R}}

\newcommand{\C}{\mathbb{C}}

\newcommand{\E}{\mathbb{E}}

\newcommand{\ztilde}{\widetilde{z}}

\newcommand{\psitilde}{\widetilde{\psi}}
\newcommand{\Psitilde}{\widetilde{\Psi}}

\newcommand{\eps}{\varepsilon}

\newcommand{\bal}{\begin{align}}
\newcommand{\eal}{\end{align}}

\newcommand{\bM}{\begin{pmatrix}}
\newcommand{\eM}{\end{pmatrix}}

\makeatletter
\newcommand{\onelettername}[1]{#1\aftergroup\@gobble}
\makeatother

\DeclareMathOperator{\ran}{ran}
\DeclareMathOperator{\Rea}{Re}

\DeclareMathOperator{\Id}{Id}

\DeclareMathOperator{\argmin}{argmin}

\DeclareMathOperator{\PSt}{{P_S^\eps}}
\DeclareMathOperator{\PAt}{{P_A^\eps}}
\DeclareMathOperator{\Pf}{P}
\DeclareMathOperator{\Pft}{{P^\eps}}

\numberwithin{equation}{section}

\newcommand{\operp}{\mathop{\bigcirc\kern-12.75pt\perp}\nolimits}

\numberwithin{equation}{section}

\allowdisplaybreaks

\definecolor{B}{rgb}{0.15,0.4,0.8}

\begin{document}
\selectlanguage{english}

\title{Stable reconstructions for the analysis formulation of $\ell^p$-minimization using redundant systems}
\author{
  Jackie Ma\footnote{ma@math.tu-berlin.de}\\
  Technische Universit\"{a}t Berlin\\ Department of Mathematics \\
  Stra\ss{}e des 17. Juni 136,   10623 Berlin
}

\maketitle

\begin{abstract}
\noindent
In compressed sensing sparse solutions are usually obtained by solving an $\ell^1$-minimization problem. Furthermore, 
the sparsity of the signal does need not be directly given. In fact, it is sufficient to have a signal that is sparse after an 
application of a suitable transform. In this paper we consider the stability of solutions obtained from $\ell^p$-minimization 
for arbitrary $0<p \leq1$. Further we suppose that the signals are sparse with respect to general redundant transforms 
associated to not necessarily tight frames. Since we are considering general frames the role of the dual frame has to be
additionally discussed. For our stability analysis we will introduce a new concept of so-called 
\emph{frames with identifiable duals}. Further, we numerically highlight a gap between the theory and the applications of 
compressed sensing for some specific redundant transforms.
\end{abstract}


\section{Introduction}

Suppose we are interested in the reconstruction of a certain object of interest but we are not enabled to observe the signal
directly. We are instead allowed to sense it. Mathematically speaking, we understand this problem as solving a system of linear 
equations
\begin{align}
	Ax =y, \label{eq:Ax=y}
\end{align}
where $x \in \R^n$ represents the object of interest, $A \in \R^{m\times n}$ is the sensing matrix and $y \in \R^m$ represents 
the resulting observations. It is of course preferred to keep the amount of acquired data small, i.e. $m$ should be small. 
However, if $m$ is (possibly) much smaller than $n$, then \eqref{eq:Ax=y} usually does not have a unique solution in general. 
This issue can be resolved by imposing additional assumptions. One of such possible assumptions is the concept of 
\emph{sparsity}. We say a signal $x$ is  \emph{sparse}, if the vector $x$ contains only very few non-zero entries. This concept 
of sparsity is fundamental in the field of \emph{compressed sensing} \cite{CanRomTao2, CanTao, Don} where one aims to recover 
sparse signals using only very few measurements. With a view to solve \eqref{eq:Ax=y} a \emph{compressed sensing reconstruction}
is then typically computed by solving an \emph{equality constrained minimization} problem of the form
\begin{align}
	\min_x f(x) \quad \text{ subject to } \quad y=Ax, \tag{$\Pf$} \label{eq:f1}
\end{align}
where $f: \R^n \to \R^+$ is some non-negative function, the so-called \emph{prior-function} or \emph{regularizer}. However, it 
is often not feasible in practice to acquire the data $y$ without any errors or noise this means instead of having 
\eqref{eq:Ax=y} we rather have
\begin{align*}
	Ax \approx y
\end{align*}
where we use $a \approx b$ to indicate $a$ is close to $b$ but not necessarily equal. Therefore, in order to allow perturbations
during the measurement process, the model \eqref{eq:f1} is extended to
\begin{align}
	\min_x f(x) \quad \text{ subject to } \quad \| y - Ax \|_2 \leq \eps, \tag{$\Pft$} \label{eq:f2}
\end{align}
where $\eps>0$ is a control parameter for the error that may arise during the acquisition. Note that any element satisfying 
the constraint in \eqref{eq:f1} also satisfies the constraint in \eqref{eq:f2}.

A major question is of course how to choose the prior $f: \R^n \to \R^+$. In the spirit of compressed sensing the ideal choice 
would be
\begin{align*}
	f (x) = \| x \|_0 := \# \{ k \, : \, x_k \neq 0 \}
\end{align*}
since for this particular choice the problem  \eqref{eq:f1} returns the sparsest solution matching the constrained. 
Unfortunately, for $f (x) = \| x \|_0$ the \emph{$\ell^0$-problem} \eqref{eq:f1} is NP-hard \cite{FouRau}. A common escape is 
a convex relaxation of the $\ell^0$-problem by using, for example
\begin{align*}
	f (x) = \| x \|_1 = \sum_{k \leq n} |x_k|
\end{align*}
in \eqref{eq:f1}, which is also known as \emph{basis pursuit} and has been investigated widely in 
\cite{CanTao, Can, CanTao2, CanRomTao2, CanRomTao1}. However, for $0<p<1$ the $\ell^p$-quasi norm 
\begin{align*}
	\|x\|_p = \left( \sum_{k \leq n} |x_k|^p \right)^{1/p}, \quad 0<p<1
\end{align*}
is closer to the $\ell^0$-function $\|\cdot\|_0$ than $\| \cdot \|_1$ is, see Figure \ref{fig:lpballs}. Therefore it is of 
natural interest to study problem \eqref{eq:f1} for $f (x) = \|x\|_p$ which has also been done, for instance, in 
\cite{Cha, SaaChaYil, LyuLinShe}.

\begin{figure}[h]
\centering\includegraphics[scale=0.4]{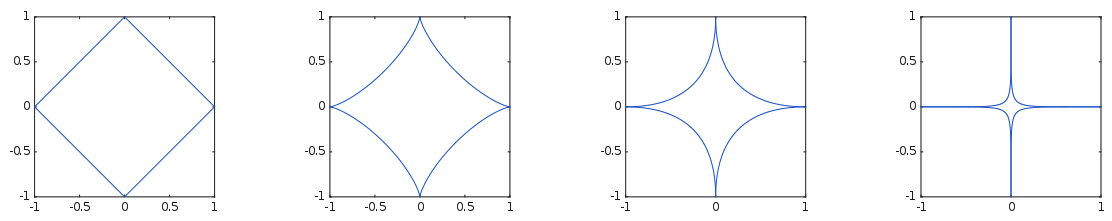}
\caption{From $\ell^1$ to $\ell^0$: Unit balls with respect to the $\ell^p$-quasi norm for decreasing $p$.}\label{fig:lpballs}
\end{figure}

%

The cases \eqref{eq:f1} and \eqref{eq:f2} with $f(x) = \|x \|_p$ and $0 \leq p \leq 1$ are usually called the 
\emph{synthesis formulation} and expect the vector $x$ to be sparse. However, it is observed in many applications, 
that the vector $x$ is only \emph{transform sparse}, i.e. sparse after an application of a suitable transform or often even 
only \emph{compressible}, i.e. very few entries are large and the rest is small in modulus but not necessarily zero.  For 
example, most natural images are \emph{compressible} with respect to multiscale systems such as wavelets \cite{Dau}, curvelets 
\cite{CanDon} or shearlets \cite{ShearletBook}. More precisely, the transform coefficients decay to zero with a certain rate as 
the scales increase, see Figure \ref{fig:SparsityLevels}. The transform sparse model is for instance used in 
\emph{magnetic resonance imaging} (MRI) where the sparsity of medical images with respect to a wavelet transform is utilized 
cf. \cite{LusDonPau}.
\begin{figure}[H]
\begin{subfigure}{0.5\textwidth}
 \includegraphics[width=0.85\textwidth]{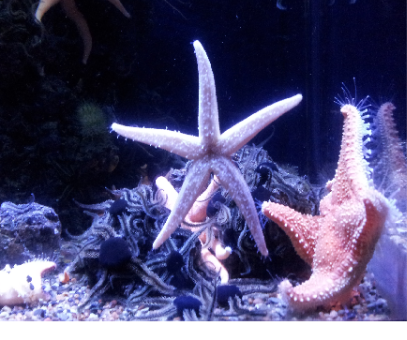}
\end{subfigure}
\begin{subfigure}{0.5\textwidth}
 \hspace*{-.8cm}\includegraphics[width=1.1\textwidth]{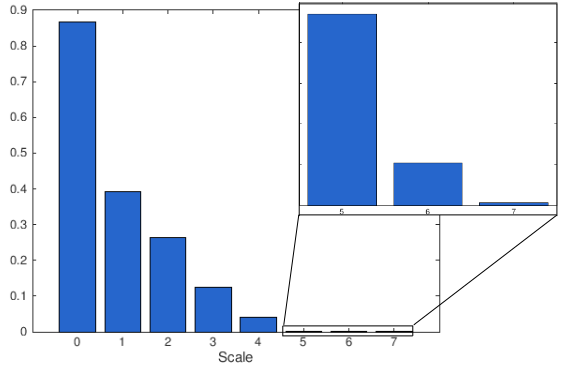}
\end{subfigure} 
\caption{\textbf{Left:} A natural image of size $2048\times2048$. \textbf{Right:} Distribution of the shearlet coefficients
of the best $N$-term approximation using 5\% of the largest coefficients in modulus.}
\label{fig:SparsityLevels}
\end{figure}
The transform sparse model leads to the so-called \emph{analysis formulation} that is problems \eqref{eq:f1} and \eqref{eq:f2} 
with
\begin{align*}
	f(x) = \| \Psi x \|_p = \left(\sum_{\lambda \leq N} |\langle x, \psi_\lambda\rangle|^p \right)^{1/p}, \quad 0<p \leq 1,
\end{align*}
where $\Psi: \R^n \to \R^N, x \mapsto (\langle x, \psi_\lambda \rangle)_{\lambda \leq N }$ is the \emph{analysis operator} of 
a spanning system $(\psi_\lambda)_{\lambda \leq N}$ for $\R^n$.

\subsection{Overview of related work}

Problems of the form \eqref{eq:f1} has been studied a lot in the literature for many different type of functions $f$.
We will now give an overview of those works that are most relevant and related to our work. We start with one of the earliest 
models in compressed sensing, that is the \emph{synthesis approach}.

\subsubsection*{Synthesis approach for $\ell^1$}

The synthesis formulation for the $\ell^1$-minimization problem is \eqref{eq:f2} with $f (\cdot )= \| \cdot \|_1$, i.e.
\begin{align}
	\min_x \| x \|_1 \quad \text{ subject to } \quad\| y - Ax \|_2 \leq \eps. \tag{$\ell^1$-$\PSt$}\label{eq:l1MinimizationSynthesis}
\end{align}
The minimization problem \eqref{eq:l1MinimizationSynthesis} is also known as \emph{(inequality constrained) basis pursuit} and 
is very well studied in compressed sensing. Stabilty results of the form
\begin{align}
	\| x^* - x\|_2 \leq C_{1} \eps + C_{2} \frac{\| x - x_s\|_1}{\sqrt{s}} \label{eq:Sterm}
\end{align}
 where $x^*$ denotes the solution to \eqref{eq:l1MinimizationSynthesis} and $x_s$ is the best $s$-term approximation, i.e. the 
 vector consisting of the $s$ largest entries of $x$, are known. This is for example the case, if $A$ satisfies the so-called 
 \emph{restricted isometriy property} (RIP) \cite{CanRomTao2, Fou, FouRau}. We recall its definition here for convenience.
\begin{deff}
Let  $A\in \C^{m \times n}$ be a  measurement matrix. If there exists $\delta_s$ such that
\begin{align}
	(1-\delta_s) \|x \|_2^2 \leq \| A x \|_2^2 \leq (1+ \delta_s)\| x\|_2^2, \label{eq:synthesisRIP}
\end{align}
for all $s$-sparse vectors $x$, then we say \emph{$A$ satisfies the RIP of order $s$} with \emph{RIP constant} $\delta_s$.
\end{deff}
Note that there are other properties used in the literature such as the \emph{null space property} \cite{CohDahDev} (and 
altered variants) to prove different stabilty estimates, see, for instance, \cite{FouRau}. 

\subsubsection*{Synthesis approach for $\ell^p$}

It has been noticed that the $\ell^p$-quasi norms yield stronger emphasis on the sparsity and thus in 
\cite{Cha, SaaChaYil,FouLai} the authors have studied the $\ell^p$-minimization problem
\begin{align}
	\min_x \| x \|_p^p \quad \text{ subject to } \quad \|y-Ax \|_2 \leq \varepsilon. \tag{$\ell^p$-$\PSt$} \label{eq:lpMinimizationSynthesis}
\end{align}
Again, the results that are fundamental for the user of this minimization problem
are stability results of the form
\begin{align}
	\| x - x^*\|_2^p \leq C_1(p) \eps^p + C_2(p) \frac{\|x-x_s\|_p^p}{s^{1-p/2}},\label{eq:Stermp}
\end{align}
where we have used the notations from above. A significant difference to the bound obtained in \eqref{eq:synthesisRIP}
is that the  constants $C_1(p)$ and $C_2(p)$ in \eqref{eq:Stermp} now depend on $p\in (0,1)$. However, for $p=1$ the 
constants in both cases agreee \cite{Cha, SaaChaYil,FouLai}. The assumption on the measurement matrix $A$ so that 
\eqref{eq:Stermp} holds is again based on the restricted isometry property.

The reader might wonder if such results are of any use in practice as the $\ell^p$-quasi norms turn the problem into a
non-convex minimization problem which are in general NP-hard \cite{DonXiaYin}, again. However, in the same work 
\cite{DonXiaYin} the authors showed that local minima can be computed in polynomial time by using an interior point method.
Further, its improvement over $\ell^1$ has also been numerically demonstrated in that paper.

\subsubsection*{Analysis approach for $\ell^1$}

Both minimization problems \eqref{eq:l1MinimizationSynthesis} and \eqref{eq:lpMinimizationSynthesis} assume the signal $x$
to be directly sparse, however, as we have already mentioned in the introduction this is often not the case in many practical
problems. One possible adaption of the synthesis problem is the analysis formulation of the minimization problem 
\eqref{eq:l1MinimizationSynthesis}, that is
\begin{align}
	\min_x \| \Psi x \|_1 \quad \text{ subject to } \quad \| y - Ax \|_2 \leq \eps, \tag{$\ell^1$-$\PAt$}\label{eq:l1MinimizationAnalysis}
\end{align}
where $\Psi$ is some sparsifying transform associated to a spanning system $(\psi_\lambda)_{\lambda \leq N}$ of $\R^n$. 
This problem has initially been studied in \cite{CanEldNeeRan} and the authors of that work proved stability results of the form
\begin{align}
	\| x^* - x\|_2 \leq C_{1} \eps + C_{2} \frac{\| \Psi x - (\Psi x)_s\|_1}{\sqrt{s}} \label{eq:StermPsi}
\end{align}
for  the case that $\Psi$ is the \emph{analysis operator} associated to a \emph{Parseval frame} and the measurement matrix $A$ 
satisfies the so-called \emph{$\Psi$-RIP}.
\begin{deff}
Let $\Psi \in \C^{N\times n}$ be a dictionary and $A\in \C^{m \times n}$ the measurement matrix. If there exists $\delta_s$ such that
\begin{align}
	(1-\delta_s) \|\Psi^* x \|_2^2 \leq \| A \Psi^* x \|_2^2 \leq (1+ \delta_s)\| \Psi^* x\|_2^2, \label{eq:RIP}
\end{align}
for all $s$-sparse vectors $x$, then we say \emph{$A$ satisfies the $\Psi$-RIP} of order $s$ with \emph{$\Psi$-RIP constant} $\delta_s$.
\end{deff}
This definition of the $\Psi$-RIP is the canonical analog to the restricted isometry property for direct sparse signals. Since
its invention an avalanche of research took place producing a wide and fruitful literature on this topic.
For instance, the relationship between the analysis formulation and the synthesis formulation has been studied in 
\cite{ElaMilRub,NamDavElaGri} as well as possible generalizations of the results of \cite{CanEldNeeRan} to general frames 
\cite{KabRauZha, Fou, KabRau, AldChePow, VaiPeyDosFad} and not only the case where $\Psi$ corresponds to to a Parseval frame. 
Such generalizations are of great importance as the they show that not all frames are equally well suited for compressed sensing.
By that we mean the number of measurements might vary siginificantly depending on the frame bounds, which is clearly not an
effect that enters the argument if one is dealing with Parseval frames. This observation already hints that the choice of the
sparsifying transform is a very delicate problem. We will next discuss what we are interested in in this work as well as what
has to be carefully considered in some practical problems in that relation.

\subsubsection*{Sparsifying transforms associated to arbitrary redundant systems}

In the above stated stability estimates \eqref{eq:Sterm}, \eqref{eq:Stermp} and \eqref{eq:StermPsi} it is desirable to have a 
fast decay of $\| x - x_s\|_1, \| x - x_s\|_p^p$ or $\| \Psi x - (\Psi x)_s\|_1$ respectively. When focusing on the latter 
expression it is implied that the transform coefficients $(\langle x, \psi_\lambda\rangle)_\lambda$ should decrease fast in 
magnitude in order to have a meaningful use of such a stability estimate. Recall that such an estimate was obtained using the 
$\Psi$-RIP. Therefore the desire of having a fast decay of the transform coefficients is needed on the one hand, however,
on the other hand we also want to have the $\Psi$-RIP to be fulfilled. The attentive reader will have noticed that there is a 
mismatch in this assumption. In fact, any signal of a given Hilbert space $\Hil$ can be written as
\begin{align}
 x  = \sum_\lambda \langle x, \psi_\lambda \rangle \psitilde_\lambda 
 = \sum_\lambda \langle x, \psitilde_\lambda \rangle \psi_\lambda, \label{eq:FrameIdentity}
\end{align}
where $(\psi_\lambda)_\lambda$ is a frame for the Hilbert space $\Hil$ and $(\psitilde_\lambda)_\lambda$ is a 
\emph{dual frame}, see \cite{Chr}. In particular, we have
\begin{align}
 x  \neq \sum_\lambda \langle x, \psi_\lambda \rangle \psi_\lambda \label{eq:NotFrameIdentity}
\end{align}
in general for non-tight frames. Equations \eqref{eq:FrameIdentity} and \eqref{eq:NotFrameIdentity} show that if we require a 
fast decay of the coefficients $(\Psi x)_s = (\langle x, \psi_\lambda \rangle)_s$, then we must nost assume the $\Psi$-RIP but 
rather the $\Psitilde$-RIP which is the $\Psi$-RIP with respect to the dual frame $(\psitilde_\lambda)_\lambda$ and not the 
primal frame $(\psi_\lambda)_\lambda$. This rather simple observation can yield a great problem in some cases. For example, it 
is not unusual that although a primal frame is known explicitly as well as a certain behaviour of the transform coefficients 
$(\langle x, \psi_\lambda \rangle)_\lambda$ but a dual $(\psitilde_\lambda)_\lambda$ is not explicitly given. Moreover, the 
theoretical behaviour of the dual coefficients can be completely different to the primal frame coefficients. The shearlet 
transform \cite{GuoKutLab2006, LLKW2007, Lim1, Lim2} is for instance a sparsifying transform that is often used in certain 
imaging applications \cite{EM1, ReiKieKin} as a sparsifying transform but an explicit formular of a dual does not exist -- for 
the case of compactly supported shearlets, as the band-limited ones form a Parseval frame.

This issue makes an assumption such as the restricted isometry property of the dual frame rather artificial. Switching the roles,
this means assuming the $\Psi$-RIP for the primal to hold, would on the other hand necessarily yield the minimization over dual 
frame coefficients. This again raises the question how these coefficients behave in general and if a dual is not known, then 
there is no point in minimizing over dual frame coefficients. 

Finally, we want to remark that if the computation of a dual is feasible, one could also minimize over all duals in order
to obtain the sparsest one. That again is the same as doing the synthesis formulation \cite{LiuMiLi}. However, we are 
interested in the analysis formulation in this paper and a discussion of the analysis versus synthesis formulation is 
not in the focus of this paper, for further interested on this matter we refer the reader to \cite{NamDavElaGri}.

\subsection{Contribution}

In this paper we combine \eqref{eq:l1MinimizationAnalysis} and \eqref{eq:lpMinimizationSynthesis} in order to obtain the 
best of both methods. While aiming for stability results we will also resolve the sparsity mismatch explained in the previous 
section that arises when arbitrary frames are considered. Note that there are two fundamentally different possibilities to 
combine the two problems \eqref{eq:l1MinimizationAnalysis} and \eqref{eq:lpMinimizationSynthesis}. Either one considers
\begin{align}
	\min_x \| \Psi x \|_p^p \quad \text{ subject to } \quad \| y-Ax\|_2 \leq \varepsilon \tag{$\ell^p$-$\PAt$} \label{eq:lpMinimizationAnalysis}
\end{align}
with a transform $\Psi$ that is associated to a fixed redundant system $(\psi_\lambda)_\lambda$ or
\begin{align}
	\min_x \| \Psitilde x \|_p^p \quad \text{ subject to } \| y-Ax\|_2 \leq \varepsilon \tag{$\ell^p$-$\widetilde{\PAt}$} \label{eq:lpDualMinimizationAnalysis}
\end{align}
which corresponds to the minimization over dual frame coefficients. The latter problem suggests to assume the $\Psi$-RIP 
with respect to the primal frame as then the sparsity pattern matches, cf. \eqref{eq:FrameIdentity}. Although stability is then
very much expected we shall show that this is indeed the case, for the sake of completeness. For \eqref{eq:lpMinimizationAnalysis}
the situation is different, although it comes without any surprise that a stability can again be obtained if one assume the 
$\Psi$-RIP with respect to the dual frame, but, we shall not do this. More importantly, we show that for certain types of 
frames the assumption of the $\Psi$-RIP with respect to the primal frame is enough even if one minimizes over the primal 
frame coefficients. The property that enables this approach are \emph{frames that have an identifiable dual}, see Section
\ref{sec:IdentifiableDuals}.

Our results can be connected to the literature as follows:
\begin{itemize}
	\item[--] For $p = 1$ our result generalizes the findings of \cite{CanEldNeeRan} to the case of arbitrary redundant 
	systems that are not necessarily a Parseval frame. However, if the system is a Parseval frame both results 
	(and constants) agree.
	\item[--] If $\Psi$ comes from a Parseval frame, then, due to the fact that all constants are given explicitly, we can 
	compare $\ell^p$-minimization with $\ell^1$-minimization. Therefore, by carefully handling the parameters we can obtain
	smaller constants as for $\ell^1$-minimization studied in \cite{CanEldNeeRan}.
	\item[--] Frames that have an identifiable dual are, coincidentally, a generalization of \emph{scalable frames} 
	\cite{KutOkoPhiTul} which have arisen in the literature from a completely different perspective.
\end{itemize}
%

Further, in practice it has been observed that the sparsity assumption might be impracticable, meaning that the sparsity $s$ of
the transform coefficients is larger than the dimension of the object $x$ that is to be reconstructed, if the transform is
truly redundant. We present a numerical experiment that shows that this is rather an issue coming from the implementation of
such transforms than the theory. This experiment in turn implies that there might a large gap between the theory and the 
applications of compressed sensing for redundant dictionaries, cf. Section \ref{sec:RedundancyCS}.

\subsection{Outline}

In Section \ref{sec:lpAnalysis} we present the new principles that we developed in order to obtain stability results for 
the analysis formulation of $\ell^p$-minimization \eqref{eq:lpMinimizationAnalysis} and the minimization over dual coefficients 
\eqref{eq:lpDualMinimizationAnalysis}. The proofs of our main results are presented in Section \ref{sec:Proofs} and Section 
\ref{section:Numerics} consists of a numerical experiment simulating the applicability and possible benefits of considering
$\ell^p$-minimization in practice. In Section \ref{sec:RedundancyCS} we discuss the role of redundancy and how it fits to the 
model of sparsity in compressed sensing for the case of the shearlet transform.

\subsubsection*{Notation}

Throughout this paper $\Psi :=(\psi_\lambda)_{\lambda \leq N} \subset \Hil$ denotes a \emph{frame} for some Hilbert space $\Hil$,
i.e. there exist (positive and finite) constants $c_1$ and $ c_2 $ such that
\begin{align*}
	c_1 \| x \|^2 \leq \sum \limits_{i=1}^N | \langle x, \psi_\lambda \rangle|^2 \leq c_2 \|x\|^2 \quad \forall \, x \in \Hil.
\end{align*}
The constants $c_1$ and $c_2$ are called \emph{lower frame bound} and \emph{upper frame bound}, respectively. If $c_1 = c_2 =1$,
then $\Psi$ is called \emph{Parseval frame}. With a slight abuse of notation we also denote the \emph{analysis operator} by
\begin{align*}
	\Psi: \Hil \to \C^N, \quad	x \mapsto (\langle x, \psi_\lambda \rangle)_{\lambda \leq N},
\end{align*}
and the synthesis operator $\Psi^*$ by
\begin{align*}
	\Psi^*: \C^N \to \Hil, \quad	(c_\lambda)_{\lambda \leq N} \mapsto \sum_{\lambda \leq N} c_\lambda \psi_\lambda .
\end{align*}
In most cases we will have $\Hil = \C^n$. If $x$ is supposed to be an image in $\C^{n\times n}$ we will consider 
$\Hil = \C^{n^2}$ the vectorization of the $n$ by $n$ images. 

Note that if the elements $\psi_1, \ldots, \psi_N$ are given as vectors in $\C^{n^2}$, then the synthesis operator is simply the matrix $\Psi^* \in \C^{n^2 \times N}$ with columns $\psi_\lambda, \lambda \leq N$.

\section{$\ell^p$-minimization: the analysis formulation}\label{sec:lpAnalysis}

This section contains the main concepts and results of this paper. We proceed by first proposing a new concept that is used 
later on to resolve the sparsity mismatch.

\subsection{Identifiable duals}\label{sec:IdentifiableDuals}

Let us recall the issue mentioned in the introduction. In the minimization problem \eqref{eq:lpMinimizationAnalysis} we are 
minimizing over primal frame coefficients $(c_\lambda)_{\lambda \leq N} = (\langle x, \psi_\lambda\rangle)_{\lambda \leq N}$ and
suppose we wish to prove stability estimates based on the $\Psi$-RIP. Then we notice that the $\Psi$-RIP is based on the 
sparsity of signals in the reconstruction system $\Psi=(\psi_\lambda)_{\lambda \leq N}$ due to the synthesis operation which is
the adjoint but not an inverse in the general frame case. More precisely, the following problem arises: The (assumed) sparse 
coefficient vector $(c_\lambda)_\lambda$ that we obtain from the minimization problem does not give rise to a sparse 
representation with respect to the reconstruction system $\Psi$ but rather to a dual $\Psitilde$, cf. \eqref{eq:FrameIdentity}
and \eqref{eq:NotFrameIdentity}. In fact, the dual frame coefficients might have a completely different sparsity pattern. In 
order to circumvent this problem we make the following definition.

\begin{deff}\label{def:IdentifiableDuals}
Let $I$ be some index set. We say a frame $\Psi = (\psi_\lambda)_{\lambda \in I}$ for $\Hil$ has an \emph{identifiable dual} if 
there exists a dual frame $\Psitilde$ such that for all $x \in \Hil$ and any $\lambda \in I$ the coefficient in modulus 
$|\langle x, \psitilde_\lambda\rangle|$ can be bounded from above and 
below by $|\langle x, \psi_\lambda \rangle|$, i.e. there exist constants $d_1, d_2>0$ such that
\begin{align}
 d_1 |\langle x, \psi_\lambda \rangle | \leq |\langle x, \psitilde_\lambda \rangle | \leq d_2 | \langle x, \psi_\lambda \rangle|
\label{eq:Identifiability}
 \end{align}
for all $x \in \Hil$ and all $\lambda \in I$.
\end{deff}

If $\Psi$ is a frame that has an identifiable dual, then this property ensures that there exists a dual such that the sparsity
of the primal frame coefficients leads to a sparse representation in a dual system. We proceed with discussing its relation to
other special frames.

It is obvious that every tight frame has an identifiable dual. More generally, every \emph{scalable frame} (cf. Definition 
\ref{def:ScalableFrames}) has an identifiable dual. 

\begin{deff}[\cite{KutOkoPhiTul}]\label{def:ScalableFrames}
 A frame $\Psi = \{ \psi_\lambda \, : \, \lambda \in \N\}$ for $\Hil$ is called \emph{scalable} if there exists scalars 
 $c_\lambda \geq 0, \lambda \in \N$ such that 
 \begin{align*}
 \{ c_\lambda \psi_\lambda \, : \, \lambda \in \N\}
 \end{align*}
 forms a Parseval frame for $\Hil$. If there exists $\delta >0$ such that $c_j > 0$, then we call 
 $\Psi$ \emph{positively scalable}.
\end{deff}

The following result is trivial to prove and we therefore omit its proof.
\begin{prop}\label{prop:IdentifiableDuals}
Every scalable frame has an identifiable dual.
\end{prop}

The question that naturally arises is whether the two concepts agree. This is indeed not the case as the following example shows.
\begin{example}
 It is easy to check that the set of vectors 
 \begin{align*}
\Psi =  \left\{ \begin{pmatrix} 
           2 \\ 1
          \end{pmatrix},
          \begin{pmatrix} 
           2\\ 2
          \end{pmatrix},
          \begin{pmatrix} 
           1 \\ 2
          \end{pmatrix}
\right\}
 \end{align*}
is frame for $\R^2$ that has an identifiable dual, e.g. 
\begin{align*}
\Psitilde =  \left\{ \begin{pmatrix} 
           2 \\ 1
          \end{pmatrix},
          \begin{pmatrix} 
           -2\\ -2
          \end{pmatrix},
          \begin{pmatrix} 
           1 \\ 2          \end{pmatrix}
\right\}
 \end{align*}
with $d_1 = d_2 = 1$ but it is neither a tight nor a scalable frame. The fact that it is not scalable follows from Theorem 3.6 
of \cite{KutOkoPhiTul}.
\end{example}

As we will also see later in this article the identifiability is actually not needed for the entire space but only to certain
elements, cf. Remark \ref{rem:IdUsage}

Next, we show how this concept relates to the \emph{restricted isometry property} and in particular the number 
of measurements that are needed in order for it to hold.

\subsection{Restricted isometry property}

The restricted isometry property was first introduced by Cand\`{e}s and Tao in \cite{CanTao2} based on the assumption that often
the signal that is to be reconstructed is sparse. It was then later generalized by Cand\`{e}s et al. in \cite{CanEldNeeRan} to 
transform sparse signals. This covers a large class of signals considered in imaging tasks, since the images considered are
often sparse in, for instance, a wavelet domain.
\begin{deff}
Let $\Psi \in \C^{N\times n}$ be a dictionary and $A\in \C^{m \times n}$ the measurement matrix. If there exists $\delta_s$ 
such that
\begin{align}
	(1-\delta_s) \|\Psi^* x \|_2^2 \leq \| A \Psi^* x \|_2^2 \leq (1+ \delta_s)\| \Psi^* x\|_2^2, \label{eq:RIP}
\end{align}
for all $s$-sparse vectors $x$, then we say \emph{$A$ satisfies the $\Psi$-RIP} of order $s$ with \emph{$\Psi$-RIP constant} 
$\delta_s$.
\end{deff}

The existence of matrices satisfying the $\Psi$-RIP is given by the following theorem which will also be needed for our main 
result. We move the proof to the appendix as it very close to the proof given in \cite{KraNeeWar}.

\begin{theorem}\label{theorem:RIPKraNeeWar2}
Fix a probability measure $\nu$ on $\{1, \ldots, N \}$, a sparsity level $s < N$, and a constant $0<\delta<1$. Let 
$\Psi=(\psi_\lambda)_{\lambda \leq N}$ be a frame with frame bounds $c_1$ and $c_2$ and let $A$ be an $n \times n$ matrix whose 
rows $(r_i)_{i \leq n}$ satisfy
\begin{align*}
	\sum_i r_{k}(i)r_j(i)\nu_i = \delta_{j,k}.
\end{align*}
Furthermore, let $K$ be a number such that $\| \psi_\lambda \| \leq K$ for all $\lambda \leq N$ and 
\begin{align}
L = \sup_{ \substack{\| \Psi^* c \| =1 \\ \| c \|_0 \leq s}} \frac{\|  (\Psi \Psi^* c)_\lambda \|_1}{\sqrt{s}}.\label{eq:LocalizationFactor}
\end{align}
 If  $\widetilde{A}$ is an $m\times n$ submatrix whose rows are subsampled from $A$ according to $\nu$, then there exists 
a $C>0$ independent of all relevant parameters such that for
\begin{align}
m &\geq \frac{CK}{c_1} \delta^{-2} s L^2 \max\{ \log^3(sL^2) \log(N), \log(1/\gamma)\} \label{eq:RipNumMeas}
\end{align}
the normalized submatrix $\sqrt{\frac{1}{m}} \widetilde{A}$ satisfies the $\Psi$-RIP of order $s$ with $\Psi$-RIP constant 
$\delta$ with probability $1-\gamma$. Furthermore, if the frame possess an identifiable dual with upper constant $d_2$, then 
$1/c_1$ can be replaced by $d_2$ in \eqref{eq:RipNumMeas}, i.e. the result holds for
\begin{align}
m &\geq {d_2 C K} \delta^{-2} s L^2 \max\{ \log^3(sL^2) \log(N), \log(1/\gamma)\} \label{eq:RipNumMeasNew}
\end{align}
\end{theorem}

Theorem \ref{theorem:RIPKraNeeWar2} states that the $\Psi$-RIP can be guaranteed provided the number of measurements
scales properly with the sparstity $s$. In fact, there are some further subtleties that have to be considered. 
First, it is not only the sparsity $s$ that is important but also the localization factor $L$. This factor can be controlled,
for instance, if the frame is \emph{localized}, \cite{ForGro}. In particular, if the Gramian has a strong off-diagonal decay, 
then $L$ can be further characterized with respect to this decay rate as $L$ measures how the Gramian distorts the sparsity
structure. Second, the dependency of the lower frame bound by $1/c_1$ can also be controlled by localization arguments.
Note that by Greshgorin's Circle Theorem there exists a $\lambda \in \{1, \ldots, N\}$ such that
\begin{align*}
\left| \frac{1}{c_1} - \|\psi_\lambda \|^2 \right| \leq \sum_{\mu \neq \lambda} |\langle \psitilde_\lambda, \psitilde_\mu \rangle |
\end{align*}
and therefore
\begin{align*}
	\frac{1}{c_1} \leq \sum_{\mu \leq N} |\langle \psitilde_\lambda, \psitilde_\mu \rangle|.
\end{align*}
Thus frames with strong localization properties such that the dual frame is also strongly localized are of particular interest.
It was also proven in \cite{ForGro} that if the primal frame $(\psi_\lambda)_{\lambda \in \N}$ is 
\emph{polynomially self-localized}, i.e. there exists $C>0$ such that
\begin{align*}
	|\langle \psi_\lambda , \psi_\mu \rangle | \leq C \frac{1}{(1 + |\lambda - \mu|)^r},
\end{align*}
for some $r \in \N$ and all $\lambda \neq \mu$, then the dual frame $(\psitilde_\lambda)_{\lambda \in \N}$ will also be 
polynomially self-localized. This result was extended in \cite{ParLoc} to more general geometries, i.e. to other index 
functions than the euclidean distance. 

 Theoretically, $s$ can be considered as the dominating value for the number of measurements. However, in practice it is a 
 priori not clear that it can always be assumed to be smaller than $n$. In fact, since $N$ is usually much much larger than 
 $n$ the analysis coefficients must be extremely sparse to make this result useful. In Section \ref{sec:RedundancyCS} 
 we will demonstrate that this is indeed an issue in many imaging problems, but we also show that it has to be further studied 
 from a practical point of view and is not really an issue of the theory.

The last concept that we need before we can prove the stability result is the notion of \emph{stable frame bounds}. This is 
merely a concept used in order to make all constants and variables feasible. It also shows how large $s$ can be in the worst
case.

\subsection{Stable frame bounds}

As we are now not necessarily dealing with tight frames or Parseval frames, the ratio of the frame bounds $c_1/c_2$ is not 
constant and in fact plays a vital role in the estimates for the proof of our main results. In order to have some control over 
this ratio we make the following definition.

\begin{deff}
We say a frame  has \emph{$q$-controllable frame bounds} if  there exists a $q \in \{1, \ldots, N\}$ such that
\begin{align}
	\frac{c_1^{p/(2-p)}}{c_2^{p/(2-p)}} \geq \frac{q}{N}. \label{eq:AssumptionFrameBounds}
\end{align}
\end{deff}
Note that \eqref{eq:AssumptionFrameBounds} is, same as the identifiability condition in Section \ref{sec:IdentifiableDuals}, 
always fulfilled for tight frames, in particular it holds for every orthonormal basis or Parseval frame. 

We will need the controllability of the frame bound in order to make the implicit assumptions on the range of the sparsity $s$ 
more precise, cf. Theorem \ref{theorem:lpMinimizationAnalysis}.

Furthermore, the controllability of the frame bounds is also of numerical importance since the number of frame elements $N$ is
usually very large in practice and in order to have numerical stability the frame ratio $c_1/c_2$ should not be too small. In 
Section \ref{section:Numerics} we will verify \eqref{eq:AssumptionFrameBounds} for \emph{digital shearlets}, \cite{ShearletBook}. 
Note that the term on the left-hand side of \ref{eq:AssumptionFrameBounds} decreases as $p$ increases.

\subsection{Stability for analysis formulation}

Based on the previously introduced principles of identifiable duals and controllable frame bounds, we can prove the following theorem that guarantees stability of solutions obtained by \eqref{eq:lpMinimizationAnalysis}. The proof of the theorem is postponed to Section \ref{sec:Proofs}. 

\begin{theorem}\label{theorem:lpMinimizationAnalysis}
Let $\Psi$ be a frame that has $q$-controllable frame bounds and has an identifiable dual. Moreover, let $A$ be a measurement matrix satisfying the $\Psi$-RIP with $\delta_{\nu} < 0.5$  where $\nu = s \frac{2\cdot (1+2^{-p})^{1/(p-2)}c_2^{p/(p-2)}d_2^{p/(p-2)}}{c_1^{p/(p-2)}}$ and  $s<\frac{q}{2d_2^2 (1+2^{-p})^{1/(p-2)}}$ and $d_2>0$ is the upper constant in the identifiability conditions. Then the solution $x^*$ of \eqref{eq:lpMinimizationAnalysis} satisfies
\begin{align*}
	\| x - x^* \|_2^p \leq C_1(p) \eps^p + C_2(p) \frac{\| \Psi x - (\Psi x)_s\|_p^p}{s^{1-p/2}}
\end{align*}
for some positive constants $C_1(p)$ and  $C_2(p)$ that depend on $p$, the frame bounds $c_1, c_2$, the $\Psi$-constant $\delta_{\nu}$, the sparsity $s$, the controllability parameter $q$, and the constants from the identifiability condition. 
\end{theorem}

\begin{rem}\label{rem:IdUsage}
The proof of Theorem \ref{theorem:RIPKraNeeWar2} and Theorem \ref{theorem:lpMinimizationAnalysis} show that we use the 
identifiability only for the rows of the measurement matrix and $x-x^*$. Therefore, it does not need to hold for every element 
of $\C^n$. Further, the dual does not have to be known or constructed explicitly. It suffices to know that it exists.
\end{rem}

For $p =1$ Theorem \ref{theorem:lpMinimizationAnalysis} generalizes the findings of \cite{CanEldNeeRan} to the case of non-tight redundant systems, thus our result can be seen as a generalization of the stability result for the analysis formulation of $\ell^1$-minimization \eqref{eq:l1MinimizationAnalysis}. In particular, if $\Psi$ is a Parseval frame the constants $C_1(p), C_2(p)$ agree for $p=1$ with those obtained in \cite{CanEldNeeRan}.  Furthermore, the constants $C_1(p), C_2(p)$ appearing in the proof are monotonically decreasing for decreasing $p$. We state this fact as a corollary.

\begin{cor}\label{corollary:lpMinimizationAnalysisParsevalFrame}
Let $\Psi$ be a Parseval frame and let $A$ be a measurement matrix satisfying the $\Psi$-RIP with $\delta_{7s} < 0.6$ and  $s<\frac{q}{2\cdot (1+2^{-p})^{1/(p-2)}}$. Then the solution $x^*$ of \eqref{eq:lpMinimizationAnalysis} satisfies
\begin{align*}
	\| x - x^* \|_2^p \leq C_1(p) \eps^p + C_2(p) \frac{\| \Psi x - (\Psi x)_s\|_p^p}{s^{1-p/2}}
\end{align*}
for some positive constants $C_1(p)$ and  $C_2(p)$ that depend on $p$. Moreover, the constants can be chosen to be monotonically decreasing for decreasing $p$ with maximal constants agreeing with those obtained in \cite{CanEldNeeRan} for $p=1$.
\end{cor}

For the sake of completeness we also discuss the minimization problem \eqref{eq:lpDualMinimizationAnalysis} while
assuming the $\Psi$-RIP.

\subsection{Stability for dual analysis formulation}\label{sec:DualAnalysisProblem}

As we already mentioned in Section \ref{sec:IdentifiableDuals}, our motivation for the concept of an identifiable dual was to 
deal with the mismatch between the sparsity system $(\psi_\lambda)_{\lambda \leq N}$ and the coefficients 
$(\langle x, \psi_\lambda \rangle)_{\lambda \leq N}$ which do not give rise to a representation in 
$(\psi_\lambda)_{\lambda \leq N}$. It is therefore very natural to expect when minimizing over the dual frame coefficients, the
concept of an identifiable dual is not needed. This is indeed the case as the following theorem shows.

\begin{theorem}\label{theorem:lpDualMinimizationAnalysis}
Let $A$ be a measurement matrix satisfying the $\Psi$-RIP with $\delta_{\nu} < 0.5$  where $\nu = s\left( \frac{3}{c_1^{2p}}\right)^{1/(2-p)}$ and $s < N\left( \frac{c_1^{2p}}{3}\right)^{1/(2-p)}$. Then the solution $x^*$ of \eqref{eq:lpDualMinimizationAnalysis} satisfies
\begin{align*}
	\| x - x^* \|_2^p \leq C_1(p) \eps^p + C_2(p) \frac{\| \Psitilde x - (\Psitilde x)_s\|_p^p}{s^{1-p/2}}
\end{align*}
for some positive constants $C_1(p)$ and  $C_2(p)$ that depend on $p$, the frame bounds $c_1,c_2$, the $\Psi$-RIP constant $\delta_\nu$, and the sparsity $s$. 
\end{theorem}

The reader might wonder why the controllability of the frame bounds is not assumed. In fact, the role of the controllability 
parameter $q$ was to determine the range of $s$. However, loosely speaking, the constants are now appearing in the right order
which makes the assumptions superfluous and the range of $s$ can be determined entirely.

The proof follows very closely the lines of the proof of Theorem \ref{theorem:lpMinimizationAnalysis}. We provide a sketch of
it in Subsection \ref{subsecProofDualProblem}.

\section{Proofs}\label{sec:Proofs}

\subsection{Proof of Theorem \ref{theorem:lpMinimizationAnalysis}}\label{subsec:ProofMain}
The proof of Theorem \ref{theorem:lpMinimizationAnalysis} is mainly inspired by \cite{CanEldNeeRan} and \cite{SaaChaYil} which
have their roots in \cite{CanRomTao2}.

Let $x, x^*$ be as in the theorem and define $z = x  - x^*$. Further, set $T_0$ as the the set consisting of the $s$ largest 
coefficients of $\Psi x$ in $p$-th modulus. For any set $T$, $\Psi_T$ should denote the matrix restricted to columns indexed 
by $T$. Then we divide $T_0^c$ into sets $T_1, T_2, \ldots$ of size $M$ in order of decreasing magnitude of $\Psi_{T_0^c}z$. 
The explicit value of $M$ will be determined later. We shall need the following results.

The first one, Lemma \ref{lemma:ConeConstraint} below, is a trivial modification of the \emph{cone constraint} Lemma 2.1 in \cite{CanEldNeeRan}. For completeness we provide a proof.

\begin{lemma}\label{lemma:ConeConstraint}
The vector $\Psi z $ obeys the following cone constrained
\begin{align*}
	\| \Psi_{T_0^c}z\|_p^p \leq 2 \| \Psi_{T_0^c} x \|_p^p + \| \Psi_{T_0} z \|_p^p.
\end{align*}
\end{lemma}
\begin{proof}
Since $x$ and $x^*$ are both feasible and $x^*$ is the minizer of \eqref{eq:lpMinimizationAnalysis} we have
\begin{align*}
	\| \Psi_{T_0} x \|_p^p+ \| \Psi_{T_0^c} x \|_p^p 
	&\geq \| \Psi x \|_p^p \\
	&\geq \| \Psi x^* \|_p^p \\
	&= \|Ê\Psi x - \Psi z\|_p^p \\
	&\geq \| \Psi_{T_0} x\|_p^p - \| \Psi_{T_0^c} x \|_p^p - \| \Psi_{T_0} z \|_p^p + \| \Psi_{T_0^c} x \|_p^p.
\end{align*}
\end{proof}

The following lemma generalizes Lemma 2.2 in \cite{CanEldNeeRan} to the case $p \in (0,1)$. For $p =1$ the results coincide.
\begin{lemma}\label{lemma:BoundingTheTail}
For $\rho = s/M$ and $\eta = 2 \|\Psi_{T_0^c}x\|_p^p/s^{1-p/2}$ we have
\begin{align*}
	\sum_{j\geq 2} \| \Psi_{T_j}z \|_2^p \leq \rho^{1-p/2}(\|\Psi_{T_0} z \|_2^p + \eta).
\end{align*}
\end{lemma}

\begin{proof}
By construction we have that every entry in $| (\Psi_{T_{j+1}}z)|^p$ which will be denoted by $| (\Psi_{T_{j+1}}z)|_{(k)}^p$ is bounded as follows
\begin{align*}
	| (\Psi_{T_{j+1}}z)|_{(k)}^p \leq \frac{\| \Psi_{T_j} z \|_p^p}{M}.
\end{align*}
Therefore
\begin{align*}
	\| \Psi_{T_{j+1}} z \|_2^2 = \sum_{k=1}^M | (\Psi_{T_{j+1}}z)|_{(k)}^2 \leq M^{1-2/p} \| \Psi_{T_j} z \|_p^2
\end{align*}
and thus
\begin{align*}
	\sum_{j\geq 2}\| \Psi_{T_{j}} z \|_2^p \leq \sum_{j\geq 1} \frac{ \| \Psi_{T_j} z \|_p^p}{M^{1-p/2}} = \frac{ \| \Psi_{T_0^c} z \|_p^p}{M^{1-p/2}}.
\end{align*}
Applying Lemma \ref{lemma:ConeConstraint} 
and the generalized mean inequality gives
\begin{align*}
	\sum_{j\geq 2}\| \Psi_{T_{j}} z \|_2^p \leq   \left( \frac{s}{M} \right)^{1-p/2}\left(\| \Psi_{T_0} z \|_2^p+ \frac{2 \|\Psi_{T_0^c} x \|_p^p}{s^{1-p/2}}\right).
\end{align*}
\end{proof}

\begin{lemma}[\cite{CanEldNeeRan}, Lemma 2.3]\label{lemma:TubeConstraint}
The vector $A z $ satisfies
\begin{align*}
	\| Az \|_2 \leq 2 \eps.
\end{align*}
\end{lemma}


\begin{lemma}\label{lemma:ConsequenceRIP}
Let $c_2$ be the upper frame bound for the frame $\Psi$. Then we have
\begin{align*}
&(1-\delta_{s+M})^{p/2} \|(\Psi_{T_{01}})^*\Psi_{T_{01}} \ztilde \|_2^p -\\
& (c_2 (1+\delta_M))^{p/2}\rho^{1-p/2} \left( c_2^{p/2} d_2^p \|  z\|_2^p + \eta\right) \leq 	(2 \eps )^p
\end{align*}
where $T_{01} := T_0 \cup T_1$. 
\end{lemma}

\begin{proof}
W.l.o.g. we assume the identifiable dual is the canonical dual and define $\widetilde{z} := (\Psi^* \Psi)^{-1}z$, note that $\Psi^* \Psi$ is invertible as it is the frame operator. Then by using Lemma \ref{lemma:TubeConstraint} we have 
\begin{align}
	(2 \eps)^p &\geq \| Az \|_2^p \\
	&= \| A \Psi^* \Psi \ztilde\|_2^p  \geq \| A (\Psi_{T_{01}})^*\Psi_{T_{01}} \ztilde \|_2^p \\
	&- \sum_{j\geq 2} \| A (\Psi_{T_{j}})^* \Psi_{T_{j}} \ztilde\|_2^p. \label{eq:ProofConsequenceRIP1}
\end{align}
 and by the $\Psi$-RIP we have
\begin{align}
(1-\delta_{s+M})^{p/2} \|(\Psi_{T_{01}})^*\Psi_{T_{01}} \ztilde \|_2^p \leq \| A (\Psi_{T_{01}})^*\Psi_{T_{01}} \ztilde \|_2^p
\label{eq:ProofConsequenceRIP2}
\end{align}
and
\begin{align}
\sum_{j\geq 2} \| A (\Psi_{T_{j}})^* \Psi_{T_{j}} \ztilde\|_2^p \leq (1+\delta_M)^{p/2} \sum_{j\geq 2} \|  (\Psi_{T_{j}})^* \Psi_{T_{j}} \ztilde\|_2^p.
\label{eq:ProofConsequenceRIP3}
\end{align}
Recall that by the identifiability there exists $d_1, d_2$ be such that
\begin{align}
	d_1 |\langle z, \psi_\lambda \rangle| \leq |\langle z, \psitilde_\lambda \rangle| \leq d_2 |\langle z, \psi_\lambda \rangle| \quad \forall \, \lambda \leq N. \label{eq:d1d2}
\end{align}
Using the frame property, the identifiability and Lemma \ref{lemma:BoundingTheTail} yields 
\begin{align}
 \sum_{j\geq 2} \|  (\Psi_{T_{j}})^* \Psi_{T_{j}} \ztilde\|_2^p &\leq  \sum_{j\geq 2}  c_2^{p/2} d_2^p  \| \Psi_{T_{j}} z \|_2^p \\
 &\leq c_2^{p/2} d_2^p \rho^{1-p/2} (\| \Psi_{T_{0}} z\|_p^p + \eta) \\
 &\leq c_2^{p/2} d_2^p  \rho^{1-p/2} (c_2^{p/2}\| z\|_p^p + \eta).
\label{eq:ProofConsequenceRIP4}
\end{align}
Combing \eqref{eq:ProofConsequenceRIP1}, \eqref{eq:ProofConsequenceRIP2}, \eqref{eq:ProofConsequenceRIP3} and \eqref{eq:ProofConsequenceRIP4} yields the claim.
\end{proof}

\begin{lemma}\label{lemma:ErrorBound}
The following inequality is true
\begin{align*}
	\|z\|_2^{2p} 
	\leq \frac{1}{c_1^{p}d_1^{2p}}\left(\frac{1}{c_1^{p/2}}\| z\|_2^{p} \| (\Psi_{T_{01}})^* \Psi_{T_{01}} \ztilde \|_2^{p} + \rho^{2-p}(c_2^{p/2}\| z \|_2^p+ \eta)^2\right),
\end{align*}
\end{lemma}
	
\begin{proof}
With $T_{01} := T_0 \cup T_1$ we have by using the frame property and the identifiability of a dual
\begin{align*}
	\|h\|_2^{2p} &\leq \frac{1}{c_1^{p}} \| \Psi z \|_2^{2p} \\
	&\leq \frac{1}{c_1^{p}d_1^{2p}}\left(\|\Psi_{T_{01}} \ztilde \|_2^{2p} + \| \Psi_{T_{01}^c} z\|_2^{2p}\right) \\
	&= \frac{1}{c_1^{p}d_1^{2p}}\left((\langle \ztilde, (\Psi_{T_{01}})^* \Psi_{T_{01}} \ztilde \rangle)^{p} + \| \Psi_{T_{01}^c}z \|_2^{2p}\right) \\
	&\leq \frac{1}{c_1^{p}d_1^{2p}}\left(\| \ztilde\|_2^{p} \| (\Psi_{T_{01}})^* \Psi_{T_{01}} \ztilde \|_2^{p} + \| \Psi_{T_{01}^c}z \|_2^{2p}\right) \\
	&\leq \frac{1}{c_1^{p}d_1^{2p}}\bigg( \frac{1}{c_1^{p/2}} \| z\|_2^{p} \| (\Psi_{T_{01}})^* \Psi_{T_{01}} \ztilde \|_2^{p}  + \rho^{2-p}(\| \Psi_{T_{0}} z \|_2^p+ \eta)^2\bigg),
\end{align*}
where the last inequality follows from Lemma \ref{lemma:BoundingTheTail}. 
\end{proof}

The proof of the main result can now be derived from the previous lemmas.
\begin{proof}[Proof of Theorem \ref{theorem:lpMinimizationAnalysis}]
By Lemma \ref{lemma:ErrorBound} we have
\begin{align*}
	\|z\|_2^{2p}
	&\leq \frac{1}{c_1^{p}d_1^{2p}}\bigg(\frac{1}{c_1^{p/2}}\left(\frac{\gamma_1 \| z\|^{2p}}{2} + \frac{ \| (\Psi_{T_{01}})^* \Psi_{T_{01}} \ztilde \|^{2p}}{2 \gamma_1}\right)+ \rho^{2-p}(c_2^{p/2}d_2^p\| z \|_2^p+ \eta)^2\bigg) \\
	&=\frac{1}{c_1^{p}d_1^{2p}}\left(\frac{1}{c_1^{p/2}}\bigg(\frac{\gamma_1 \| z\|^{2p}}{2} + \frac{ \| (\Psi_{T_{01}})^* \Psi_{T_{01}} \ztilde \|^{2p}}{2 \gamma_1}\right)   +\rho^{2-p}\left(c_2^{p}d_2^{2p}\| z \|_2^{2p}+ 2c_2^{p/2}d_2^p \| z \|_2^p\eta + \eta^2\right)\bigg) \\
	&\leq \frac{1}{c_1^{p}d_1^{2p}}\bigg(\frac{1}{c_1^{p/2}}\left(\frac{\gamma_1 \| z\|^{2p}}{2} + \frac{ \| (\Psi_{T_{01}})^* \Psi_{T_{01}} \ztilde \|^{2p}}{2 \gamma_1}\right)   +\rho^{2-p}\left(c_2^{p}d_2^{2p}\| z \|_2^{2p}+ \left(c_2^{p}d_2^{2p}\gamma_2 \| z \|_2^{2p} + \frac{\eta^2}{\gamma_2}\right) + \eta^2\right) \bigg)\\
	&= \frac{1}{c_1^{p}d_1^{2p}}\bigg(\frac{1}{c_1^{p/2}}\left(\frac{\gamma_1 \| z\|^{2p}}{2} + \frac{ \| (\Psi_{T_{01}})^* \Psi_{T_{01}} \ztilde \|^{2p}}{2 \gamma_1}\right)  +\rho^{2-p}\left( c_2^{p} d_2^{2p}(1 + \gamma_2 )\| z \|_2^{2p} + \left(\frac{1}{\gamma_2} + 1\right) \eta^2\right) \bigg).
\end{align*} 
Therefore
\begin{align*}
\left(1 - \frac{1}{c_1^pd_1^{2p}}\left( \frac{\gamma_1}{2c_1^{p/2}} + \rho^{2-p} c_2^pd_2^{2p}(1+\gamma_2)\right) \right)\|z\|_2^{2p} 
\leq \frac{1}{c_1^pd_1^{2p}}\left(\frac{\| (\Psi_{T_{01}})^* \Psi_{T_{01}} \ztilde \|^{2p}}{2\gamma_1c_1^{p/2}}+ \rho^{2-p}\left(\frac{1}{\gamma_2} + 1\right) \eta^2\right).
\end{align*}
and in particular
\begin{align*}
\left(c_1^pd_1^{2p} - \left( \frac{\gamma_1}{2c_1^{p/2}} + \rho^{2-p} c_2^pd_2^{2p}(1+\gamma_2)\right) \right)^{1/2}\|z\|_2^{p}  \leq \frac{\| (\Psi_{T_{01}})^* \Psi_{T_{01}} \ztilde \|^{p}}{\sqrt{2\gamma_1c_1^{p/2}}}+ \rho^{1-p/2}\left(\frac{1}{\gamma_2} + 1\right)^{1/2} \eta.
\end{align*}
and hence, by using Lemma \ref{lemma:ConsequenceRIP} we obtain
\begin{align*}
\Gamma_1(p) \|z\|_2^p - \Gamma_2(p) \eta \leq (2 \eps )^p 
\end{align*}
with
\begin{align*}
\Gamma_1(p) = \sqrt{2\gamma_1c_1^{p/2}(1-\delta_{s+M})^{p}(c_1^pd_1^{2p} - ( \frac{\gamma_1}{2c_1^{p/2}} + \rho^{2-p} c_2^pd_2^{2p}(1+\gamma_2)) )}  -\sqrt{(c_2 d_2^2(1+\delta_M))^{p}\rho^{2-p}} 
\end{align*}
and
\begin{align*}
\Gamma_2(p)  =   \sqrt{2\gamma_1c_1^{p/2}(1-\delta_{s+M})^{p} \rho^{2-p}\left(\frac{1}{\gamma_2} + 1\right)} -\sqrt{(c_2d_2^2 (1+\delta_M))^{p}\rho^{2-p}} .
\end{align*}
It is left to choose $\gamma_1, \gamma_2,$ and $\rho$ so that $\Gamma_1(p)$ and $\Gamma_2(p)$ are non-negative. 
W.l.o.g. we can assume $c_1 \geq 1$, otherwise a rescaling can be performed. We now choose $\gamma_1 = c_1^{3p/2}d_1^{2p}$ and 
$M = s\frac{2\cdot (1+2^{-p})^{1/(p-2)}c_2^{p/(p-2)}d_2^{2p/(p-2)}}{c_1^{p/(p-2)}}$. Then $\Gamma_1(p)$ is larger than zero and 
for $\gamma_2$ sufficiently small $\Gamma_2(p)$ is also larger than zero. Furthermore, one can check
\begin{align*}
	\Gamma_1'(p)<0 \quad \text{and} \quad \Gamma_2'(p) >0
\end{align*}
holds for $\gamma_2$ small enough. This completes the proof.

\end{proof}

\subsection{Proof of Corollary \ref{corollary:lpMinimizationAnalysisParsevalFrame}}

The Proof of Corollary \ref{corollary:lpMinimizationAnalysisParsevalFrame} follows the same lines as the Proof of Theorem 
\ref{theorem:lpMinimizationAnalysis} and we skip the details. We only have to verify the statement for regarding the constants. 
For $c_1 = c_2 =1$ the above computations reduce to 
\begin{align*}
\Gamma_1(p) \|z\|_2^p - \Gamma_2(p) \eta \leq (2 \eps )^p 
\end{align*}
with the constants
\begin{align*}
\Gamma_1(p) &= \sqrt{2\gamma_1(1-\delta_{s+M})^{p}\left( 1- \left( \frac{\gamma_1}{2} + \rho^{2-p} (1+\gamma_2)\right) \right)} -\sqrt{( (1+\delta_M))^{p}\rho^{2-p}}  \\
\Gamma_2(p) & =   \sqrt{2\gamma_1(1-\delta_{s+M})^{p} \rho^{2-p}\left(\frac{1}{\gamma_2} + 1\right)} -\sqrt{( (1+\delta_M))^{p}\rho^{2-p}} .
\end{align*}
One possibility to guarantee $\Gamma_1(p)$ and $\Gamma_2(p)$ to be positive is to choose $\gamma_1 = 1, M = 6s$ then $\Gamma_1(p)$ is positive and for $\gamma_2$ sufficiently small  $\Gamma_2(p)$ is positive. 
Moreover, note that
\begin{align*}
	\Gamma_1'(p)<0 \quad \text{and} \quad \Gamma_2'(p) >0
\end{align*}
for $\gamma_2$ small enough.

\subsection{Proof of Theorem \ref{theorem:lpDualMinimizationAnalysis}}\label{subsecProofDualProblem}

The proof follows the argumentation shown in Section \ref{subsec:ProofMain}. We will summarize the results that are needed and 
sketch a proof of the result.

\begin{prop}\label{prop:aux}
Retaining the notations and definition Subsection \ref{subsec:ProofMain} the following estimates hold
\begin{itemize}
	\item[i)]
	\begin{align*}
		\| \Psitilde_{T_0^c}z\|_p^p \leq 2 \| \Psitilde_{T_0^c} x \|_p^p + \| \Psitilde_{T_0} z \|_p^p.
	\end{align*}
	\item[ii)]
	\begin{align*}
	\sum_{j\geq 2} \| \Psitilde_{T_j}z \|_2^p \leq \rho^{1-p/2}(\|\Psitilde_{T_0} z \|_2^p + \eta).
	\end{align*}
	\item[iii)]
	\begin{align*}
	(1-\delta_{s+M})^{p/2} \|(\Psi_{T_{01}})^*\Psi_{T_{01}} \ztilde \|_2^p - (c_2 (1+\delta_M))^{p/2}\rho^{1-p/2} \left( c_1^{-p/2} \|  z\|_2^p + \eta\right) \leq 	(2 \eps )^p
	\end{align*}
	\item[iv)]
	\begin{align*}
		\|z\|_2^{2p} \leq c_2^{p}\left(\| z\|_2^{p} \| (\Psitilde_{T_{01}})^* \Psitilde_{T_{01}} \ztilde \|_2^{p} + \rho^{2-p}(\| \Psitilde_{T_0} z \|_2^p+ \eta)^2\right).
	\end{align*}
\end{itemize}
\end{prop} 

\begin{proof}
All properties are proved in the same manner as in Subsection \ref{subsec:ProofMain}, in particular, item \textit{i)}, \textit{ii)}, and \textit{iv)} are trivial adaptions and will be skipped. As for item \textit{iii)}, note that
\begin{align*}
(2\eps)^p &\geq \| Az \|_2^p \\
&\geq \| A (\Psi_{T_{01}})^*\Psi_{T_{01}} \ztilde \|_2^p - \sum_{j \geq 2} \| A (\Psi_{T_{j}})^*\Psi_{T_{j}} \ztilde \|_2^p \\
&\geq  (1- \delta_{s+M})^{p/2}\| (\Psi_{T_{01}})^*\Psi_{T_{01}} \ztilde \|_2^p  -(1+ \delta_M)^{p/2}  \sum_{j \geq 2} \|  (\Psi_{T_{j}})^*\Psi_{T_{j}} \ztilde \|_2^p.
\end{align*}
By \textit{ii)} we have
\begin{align*}
\sum_{j \geq 2} \| (\Psi_{T_j})^* \Psi_{T_j} \ztilde \|_2^p &\leq c_2^{p/2} \sum_{j \geq 2} \| \Psi_{T_j} \ztilde \|_2^p \\
&=c_2^{p/2} \sum_{j \geq 2} \| \Psitilde_{T_j} z \|_2^p \\
&\leq c_2^{p/2} \rho^{1- p/2}(\| \Psitilde_{T_0} z \|_2^p + \eta).
\end{align*}
Thus the claim follows.
\end{proof}

\begin{proof}[Sketch of Proof of Theorem \ref{theorem:lpDualMinimizationAnalysis}]
By Proposition \ref{prop:aux} \textit{iv)} we have
\begin{align*}
	\|z \|_2^{2p} &\leq c_2^p \bigg[ \left( \frac{\gamma_1 \| z \|_2^{2p}}{2} + \frac{ \| (\Psitilde_{T_{01}})^*\Psitilde_{T_{01} z \|_2^p}}{2\gamma_1}\right)  + \rho^{2-p} ( \|\Psitilde_{T_{01}} z \|_2^{2p} + 2 \| \Psitilde_{T_{0}}z \|_2^p \eta + \eta^2 \bigg] \\
&\leq c_2^p \bigg[ \left( \frac{\gamma_1 \| z \|_2^{2p}}{2} + \frac{ \| (\Psitilde_{T_{01}})^*\Psitilde_{T_{01} z \|_2^p}}{2\gamma_1}\right)   +\rho^{2-p} \left( \|\Psitilde_{T_{01}} z \|_2^{2p}  + \eta^2 + \left(\gamma_2 \| \Psitilde_{T_{0}}z \|_2^{2p} + \frac{\eta^2}{\gamma_2}\right)\right)\bigg] \\
&\leq c_2^p \bigg[ \left( \frac{\gamma_1 \| z \|_2^{2p}}{2} + \frac{ \| (\Psitilde_{T_{01}})^*\Psitilde_{T_{01} z \|_2^p}}{2\gamma_1}\right)  + \rho^{2-p} \left(\frac{1}{c_1^{2p}}(1+\gamma_2) \| z \|_2^{2p}  + \left(1 + \frac{1}{\gamma_2}\right)\eta^2 \right)\bigg].
\end{align*}
Therefore
\begin{align*}
\left( 1- \frac{c_2^p \gamma_1}{2} - \frac{\rho^{2-p}(1+\gamma_2)}{c_1^{2p}}\right) \|z \|_2^{2p}  \leq \frac{c_2^p}{2 \gamma_1} \| (\Psitilde_{T_{01}})^*\Psitilde_{T_{01}} z \|_2^{2p} + \rho^{2-p} \left( 1 + \frac{1}{\gamma_2}\right) \eta^2
\end{align*}
which in turn implies 
\begin{align*}
	\sqrt{\frac{2 \gamma_1}{c_2^p}}\bigg[\left( 1- \frac{c_2^p \gamma_1}{2} - \frac{\rho^{2-p}(1+\gamma_2)}{c_1^{2p}} \right) \| z \|_2^p -\rho^{1-p/2}\left(1 + \frac{1}{\gamma_2}\right)^{1/2} \eta \bigg] \leq \| (\Psitilde_{T_{01}})^*\Psitilde_{T_{01}} z \|_2^{p}.
\end{align*}
Using Proposition \ref{prop:aux} \textit{iii)} we conclude
\begin{align*}
	 \Gamma_1(p) \| z \|_2^p- \Gamma_2(p)\eta \leq (2\eps)^p
\end{align*}
where
\begin{align*}
\Gamma_1(p) = \sqrt{\frac{2 \gamma_1}{c_2^p} \left( 1- \frac{c_2^p \gamma_1}{2} - \frac{\rho^{2-p}(1+\gamma_2)}{c_1^{2p}}\right)( 1-\delta_{s+M})^p} -\sqrt{(1+\delta_M)^p\left(\frac{c_2}{c_1}\right)^p\rho^{2-p}}
\end{align*}
and
\begin{align*}
\Gamma_2(p) = \sqrt{(1- \delta_{s+M})^p \rho^{2-p} \left(1 + \frac{1}{\gamma_2}\right)\frac{2\gamma_1}{c_1^p}}- \sqrt{\rho^{2-p}\left(\frac{c_2}{c_1}\right)^p(1+\delta_M)^p)}.
\end{align*}
Thus choosing, for example, $\gamma_1 \leq c_1^{-p}$ and $M = s \left(\frac{2}{c_1^{2p}}\right)^{1/(2-p)}$ yields the result. Note that it is possible to choose to choose the constants such that they obey
\begin{align*}
	\Gamma_1'(p)<0 \quad \text{and} \quad \Gamma_2'(p) >0.
\end{align*}
\end{proof}

\section{Numerics}\label{section:Numerics}

The non-convex minimization problem \eqref{eq:lpMinimizationAnalysis} is in general NP-hard although interior point methods
exist to find local minimima, \cite{DonXiaYin}. The are also other algorithms proposed to solve or rather approximate
non-convex $\ell^p$-minimization problem and comparison of some methods can be found, for instance, in \cite{LyuLinShe}. 
Many of these methods, are based on \emph{iterative reweighting} which originally stems from \cite{CanWakBoy} and has been used 
to strengthen the effect of sparsity. The ideas can be transferred  to solve the analysis-based $\ell^p$-minimization problem. 
In the next section we present the algorithm that we use to solve \eqref{eq:lpMinimizationAnalysis}. It can also be found 
in \cite{FouLai} for the synthesis formulation.

\subsection{Approximating the $\ell^p$-problem by reweighting} \label{Subsec:Algorithm}

In order to solve the constrained $\ell^p$-analysis minimization problem \eqref{eq:lpMinimizationAnalysis} under the presence 
of noise we use an algorithm that is based on the following intuitive discussion.
Solving \eqref{eq:lpMinimizationAnalysis}, that is solving
\begin{align*}
	\min_x \| \Psi x \|_p^p \quad \text{ s.t. } \quad \|y- Ax\|_2 \leq \varepsilon 
\end{align*}
is equivalent to solving
\begin{align*}
	\min_x \sum_{\lambda} \frac{| \Psi x |_\lambda}{| \Psi x |_\lambda^{1-p}}  \quad \text{ s.t. } \quad \|y - Ax\|_2 \leq \varepsilon,
\end{align*}
where we neglect for a moment the fact that dividing by $| \Psi x |_\lambda$ is by no means always mathematically legit.
Using the idea of reweighted $\ell^1$-minimization \cite{CanWakBoy} we wish to strengthen the effect of sparsity by using the 
exact signal $x_0$ to consider the weighted $\ell^1$ problem
\begin{align*}
	\min_x \sum_{\lambda} \frac{| \Psi x |_\lambda}{| \Psi x_0 |_\lambda^{1-p}}  \quad \text{ s.t. } \quad \|y - Ax \|_2 \leq \varepsilon .
\end{align*}
If $(\Psi x_0)_\lambda =0$, the weight is set as $\infty$. Clearly, multiplying the denominator by some $\mu>0$ does not change 
the minimizer, hence, we consider
\begin{align*}
	\min_x \sum_{\lambda} \frac{| \Psi x |_\lambda}{(\mu | \Psi x_0 |_\lambda)^{1-p}}  \quad \text{ s.t. } \quad \|y - Ax\|_2 \leq \varepsilon .
\end{align*}
The purpose of $\mu$ is  to have better numerical control over the magnitude of the analysis coefficients and should be signal 
dependent. However, since $x_0$ is not available in advance we might consider using an approximate vector $x^k$ and consider
\begin{align}
	\min_x \sum_{\lambda} \frac{| \Psi x |_\lambda}{(\mu| \Psi x^k |_\lambda)^{1-p}}  \quad \text{ s.t. } \quad \|y - Ax\|_2 \leq \varepsilon .\label{problem:lpnumerical0}
\end{align}
The approximate vector $x^k$ can be obtained by solving \eqref{problem:lpnumerical0} after an initialization of $x_0$. Finally, 
in order to prevent any instabilities one can introduce $\nu>0$ and solve
\begin{align}
	\min_x \sum_{\lambda} \frac{| \Psi x |_\lambda}{(\mu| \Psi x^k |_\lambda + \nu)^{1-p}}  \quad \text{ s.t. } \quad \|y - Ax\|_2 \leq \varepsilon.\label{problem:lpnumerical}
\end{align}
The final program is then the following.

\begin{algorithm}
\SetKwData{Left}{left}\SetKwData{This}{this}\SetKwData{Up}{up}
\SetKwFunction{Union}{Union}\SetKwFunction{FindCompress}{FindCompress}
\SetKwInOut{Input}{Input}\SetKwInOut{Output}{Output}
 \Input{ $y, \mu,  \varepsilon,M$.}
\Output{$x$}
Initialize $k=0, W = (w_{\lambda})_{\lambda} = 1$.

 \While{$ k \leq M$}{
Find solution of \eqref{problem:lpnumerical}, i.e. find $x^k$ such that
\begin{align}
	x^k = \argmin_x \| W \Psi x \|_1 \quad \text{ s.t. } \quad \| y - Ax\|_2 \leq \varepsilon. \label{NESTA}
\end{align}
Update $W$ by setting
\begin{align*}
	W_\lambda = \frac{1}{(\mu |\Psi x^k|_\lambda + \nu)^{1-p}}.
\end{align*}
Increase $k \to k+1$.
 }
\caption{Algorithm used to solve the minimization problem \eqref{eq:lpMinimizationAnalysis}}\label{Algorithm}
\end{algorithm}

Note that \eqref{problem:lpnumerical} and \eqref{NESTA},respectively, can be solved using many different common $\ell^1$-solvers
such as NESTA (\cite{Nesta}) which is available at
\begin{center}
\texttt{http://statweb.stanford.edu/\~{}candes/nesta/}
\end{center}

\subsection{Reconstruction from Fourier measurements using shearlets}

Shearlets were first introduced in \cite{GuoKutLab2006, LLKW2007} as a frame for the Hilbert space $L^2(\R^2)$ that are based on 
\emph{anisotropic scaling}, \emph{shearing}, and \emph{translations}. It construction is not limited to the two dimensional case
and the higher dimensional case has been considered in \cite{KLemLim4}. We shall not go into detail of the theoretical
concepts of shearlets but refer the reader to \cite{ShearletBook, KitKutLim, Lim2, KutLimRei} for theoretical background an 
implementations of such systems.

The shearlet implementation used in this paper are downloaded from
\begin{center}
\texttt{http://www.shearlab.org}
\end{center}

The object of interest is the GLPU phantom introduced in \cite{GLPU} which can be downloaded at
\begin{center}
\texttt{http://bigwww.epfl.ch/algorithms/mriphantom/\#soft}
\end{center}

As already mentioned in the introduction, one of the applications of the analysis-based $\ell^p$-minimization problem is
magnetic resonance imaging where the sampling process is modeled by taking samples of the Fourier transform of the signal.
The minimization problem for this application is
\begin{align}
	\min_x \| \Psi x \|_p^p \text{ s.t. } \quad \| y - \mathcal{F}x \|_2 \leq \eps \label{MRI}
\end{align}
where $\Psi$ is the shearlet transform and $\mathcal{F}$ is the undersampled Fourier transform restricted onto a certain subset
of frequencies. Typical choices for the set of measured frequencies are samples that lie on continuous trajectories such as
horizontal lines, spirals, radial lines etc. For our numerical experiments we have chosen a radial sampling mask consisting of 
30 radial lines, see Figure \ref{fig:RMGLPU} below.

\begin{figure}[H]
\centering
	\includegraphics[width=.32\textwidth]{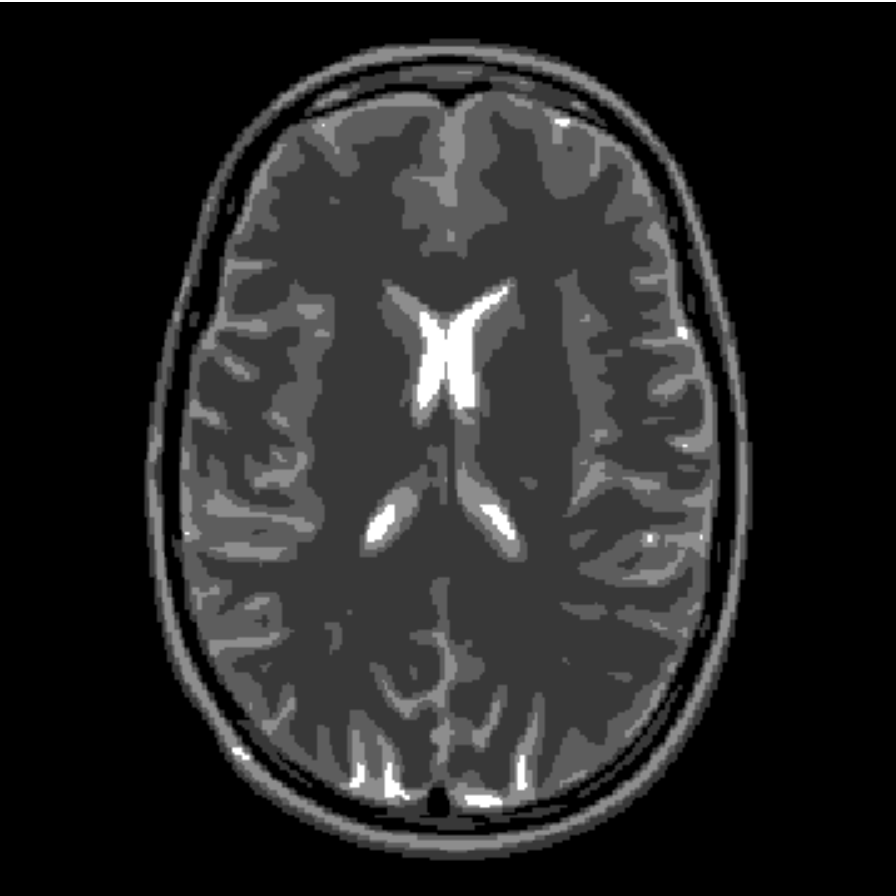}
	\includegraphics[width=.32\textwidth]{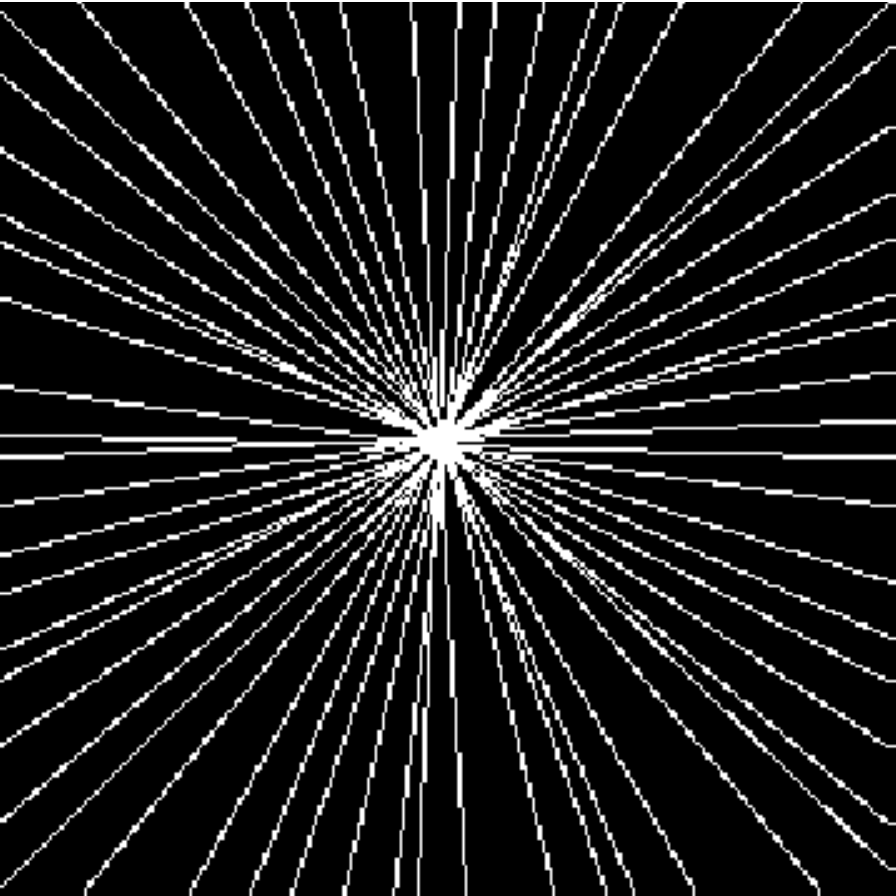}
	\includegraphics[width=.32\textwidth]{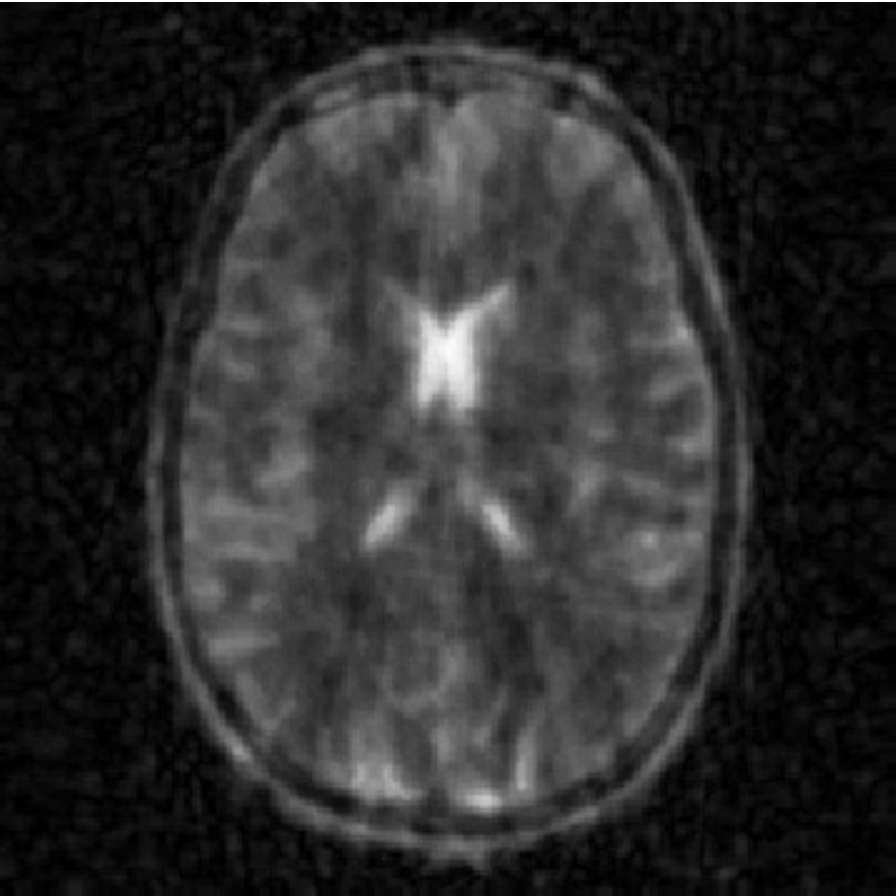}
\caption{\textbf{Left:} GLPU Phantom from \cite{GLPU}. \textbf{Middle:} Radial sampling pattern used in Fourier domain.
\textbf{Right:} Fourier inversion of data.} \label{fig:RMGLPU}
\end{figure}

We now let Algorithm \ref{Algorithm} run for a maximum of 10 iterations, i.e. $M =10$ in Algorithm \ref{Algorithm}. The
outcome of the relative error is depicted in Figure \ref{fig:RMGLPURE}. The curve for $p=1$ is not shown as it would be 
a constant line at the first error($\approx$ 17 \%). Otherwise, for decreasing $p$ we see a decrease with every increasing 
iteration $k$ of the error curves as $p$ get smaller. It is important to mention, that for $p=1$ the result corresponds to 
conventional compressed sensing, i.e. standard $\ell^1$-minimization. 

\begin{figure}[H]
\centering
	\includegraphics[width=.75\textwidth]{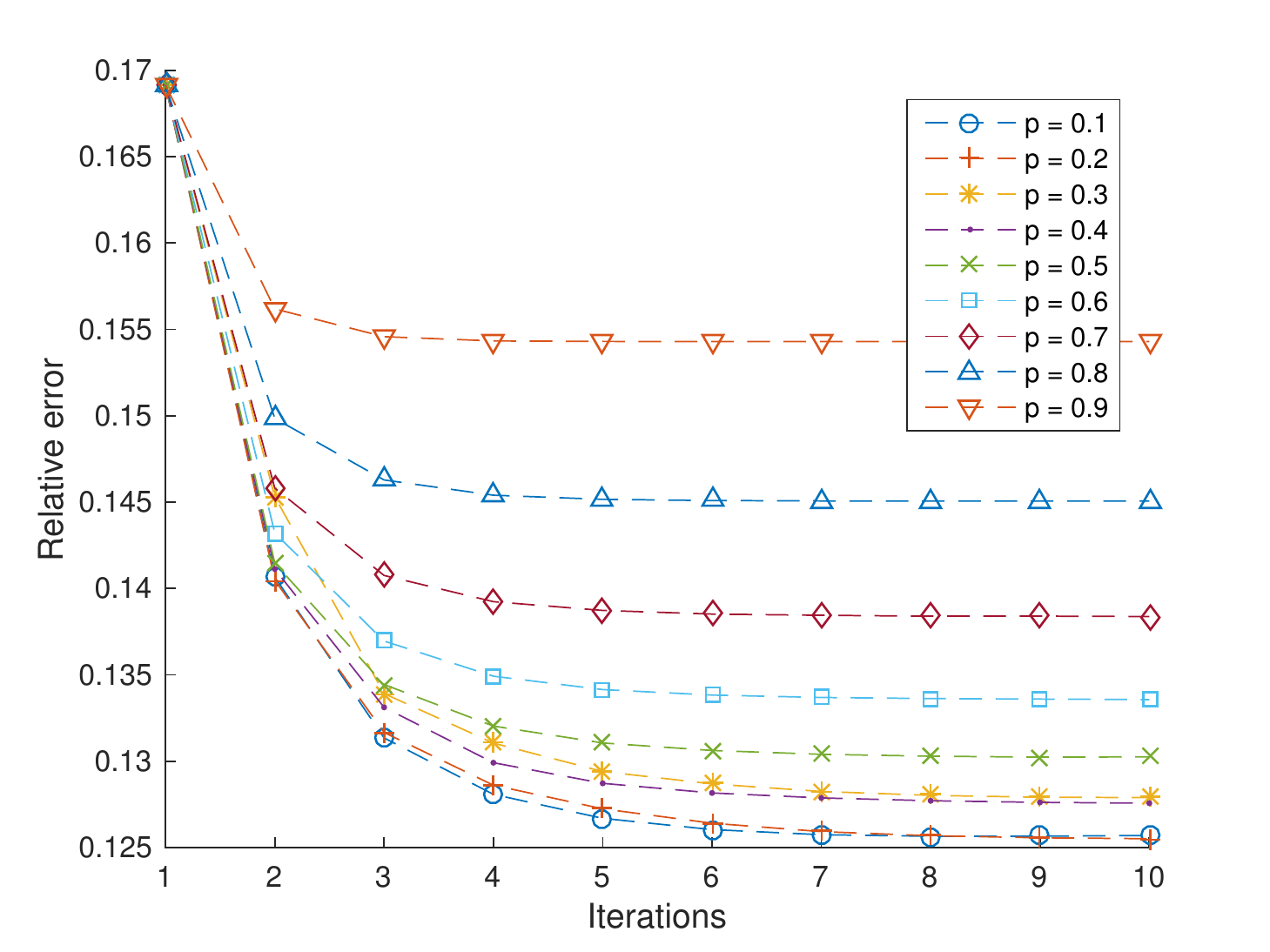}
\caption{Relative error.} \label{fig:RMGLPURE}
\end{figure}

\begin{figure}[H]
\centering
	\includegraphics[width=.49\textwidth]{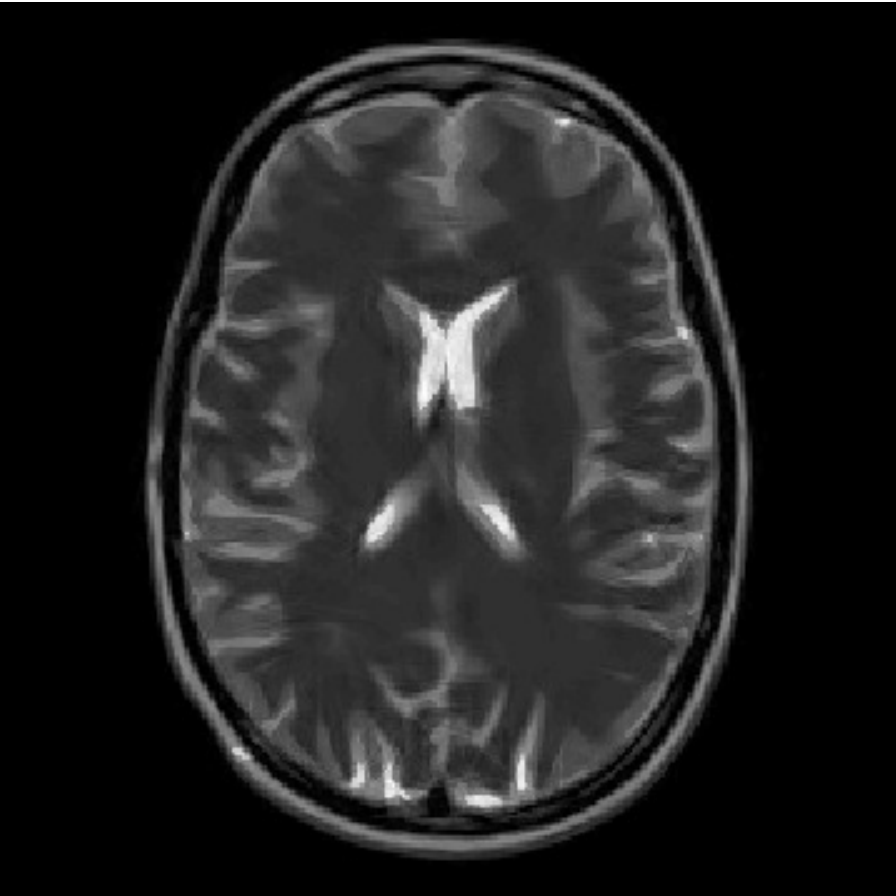}
	\includegraphics[width=.49\textwidth]{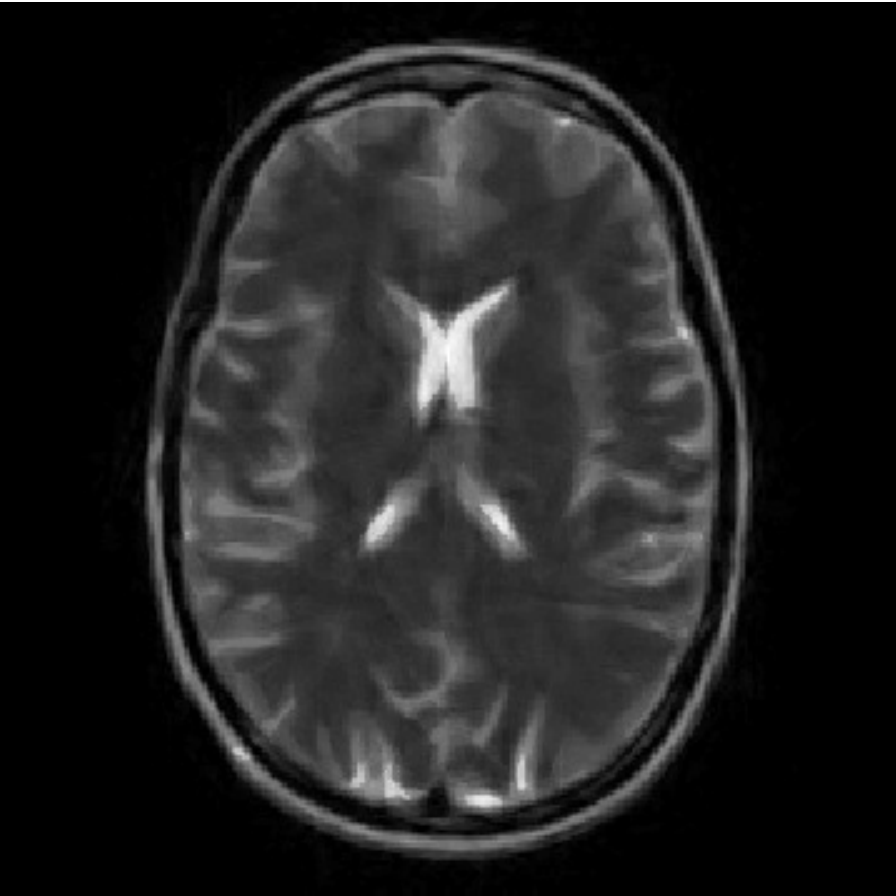}
\caption{\textbf{Left:} Reconstruction for $p = 0.1$. \textbf{Right:}  Reconstruction for $p = 1$.} \label{fig:RMGLPUR}
\end{figure}

Moreover, we like to mention that the other solutions are not obtained by using additional information except the computed
sparsity structure that is known computed from the previous $\ell^1$-minimization. More precisely, the $\ell^1$-solution is used
to obtain new weights which are now a good guess for the sparsity pattern. These weights are then returned to the algorithm so 
that it computes a new solution. In that sense the improvement is for free. However, It is computationally much more demanding
as we have to solve the minimization problem $M=10$ times. Further, the algorithm takes longer as $p$ gets smaller due to 
numerical instabilities,
cf. Figure \ref{fig:Wconv1}. 

In Figure \ref{fig:Wconv1} we plot the relative error for a different number of iterations. 
Indeed, the first $3 \times 3$ block plots show the nine sequences
\begin{align*}
 \left( \frac{ \| \Psi x^{k+1}_p - \Psi x^k_p \|}{\| \Psi x^{k+1}_p\|}\right)_{k = 1, \ldots, 10} \text{ for } p = 0.1, 0.2, \ldots, 0.9, 
\end{align*}
where
\begin{align*}
	x^{k+1}_p = \argmin_x \| W \Psi x \|_1 \quad \text{ subject to } \quad \| y - Ax\|_2 \leq \varepsilon
\end{align*}
with
\begin{align*}
	W_\lambda = \frac{1}{(\mu |\Psi x^k|_\lambda + \nu)^{1-p}}.
\end{align*}
That is, we show the outcome of Algorithm \ref{Algorithm} for all $p = 0.1, 0.2, \ldots, 0.9$ and increasing $k$ from 1 to 10. 
The second $3 \times 3$ block then shows the same sequence but starting from $k =2$ for better visibility of the error curves
\begin{align*}
	 \left( \frac{ \| \Psi x^{k+1}_p - \Psi x^k_p \|}{\| \Psi x^{k+1}_p\|}\right)_{k = 2, \ldots, 10} \text{ for }   p = 0.1, 0.2, \ldots, 0.9, 
\end{align*}
and the last one shows
\begin{align*}
	 \left( \frac{ \| \Psi x^{k+1}_p - \Psi x^k_p \|}{\| \Psi x^{k+1}_p\|}\right)_{k = 3, \ldots, 10} \text{ for }   p = 0.1, 0.2, \ldots, 0.9.
\end{align*}
By plotting the same sequence with $k\geq 2$ and $k \geq 3$ we can observe that the relative error decreases very quickly to 
zero after very few iterations which suggest fast numerical convergence of the algorithm. However, it also shows that such an 
approach via reweighting yields to less stable outcomes for very small $p$ which is strongly visible in the third $3 \times 3$ 
block plot for $p = 0.1$ and larger $k$. The oscillations also show that the choice of a stopping criterion has to be done 
very carefully. We haven't incorporated any but the maximum number of iterations, which yield a termination if this number is 
reached.

\begin{figure}[H]
\centering
\includegraphics[height=.32\textheight, width = .85\textwidth]{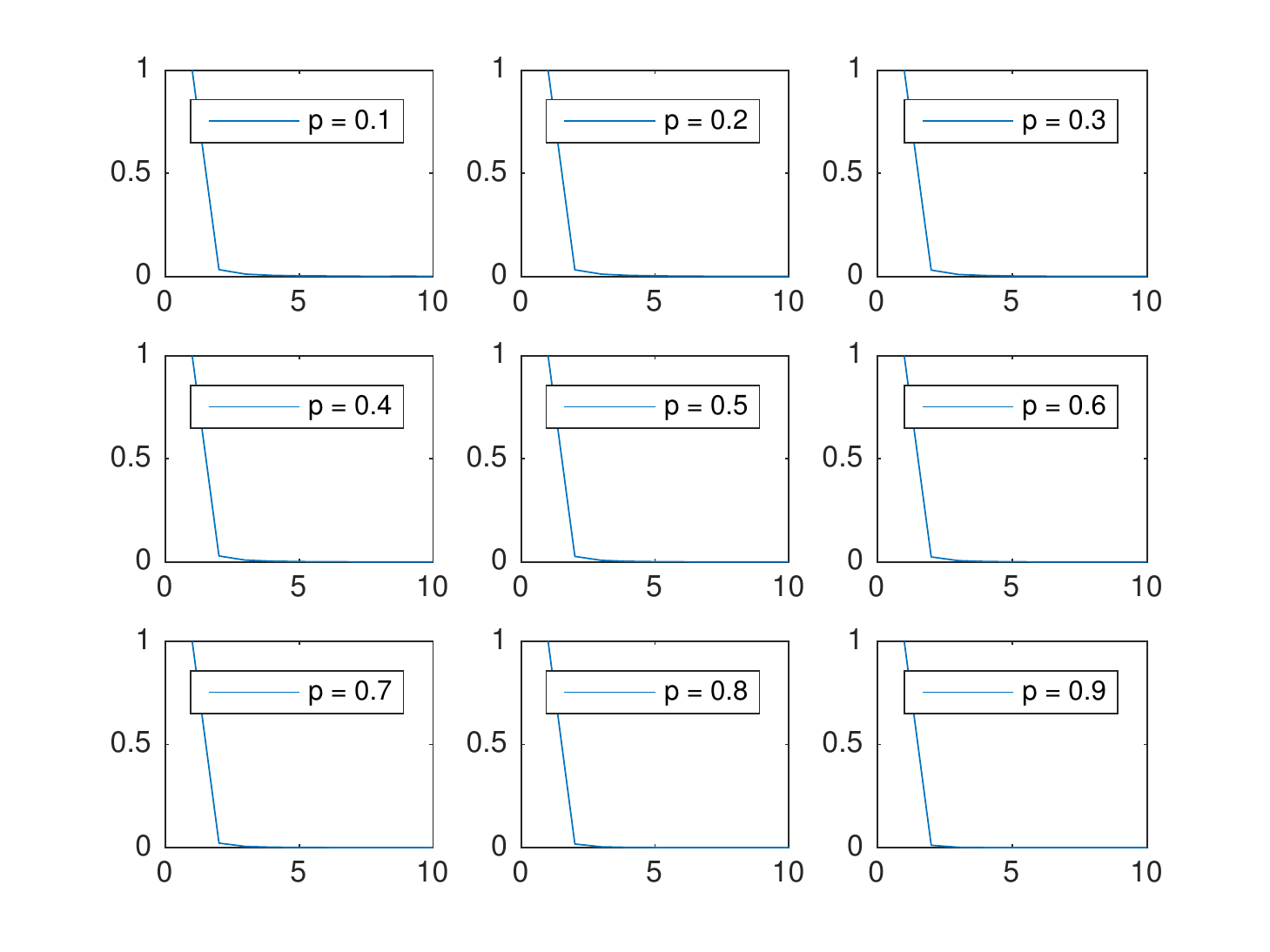}

\includegraphics[height=.32\textheight, width = .85\textwidth]{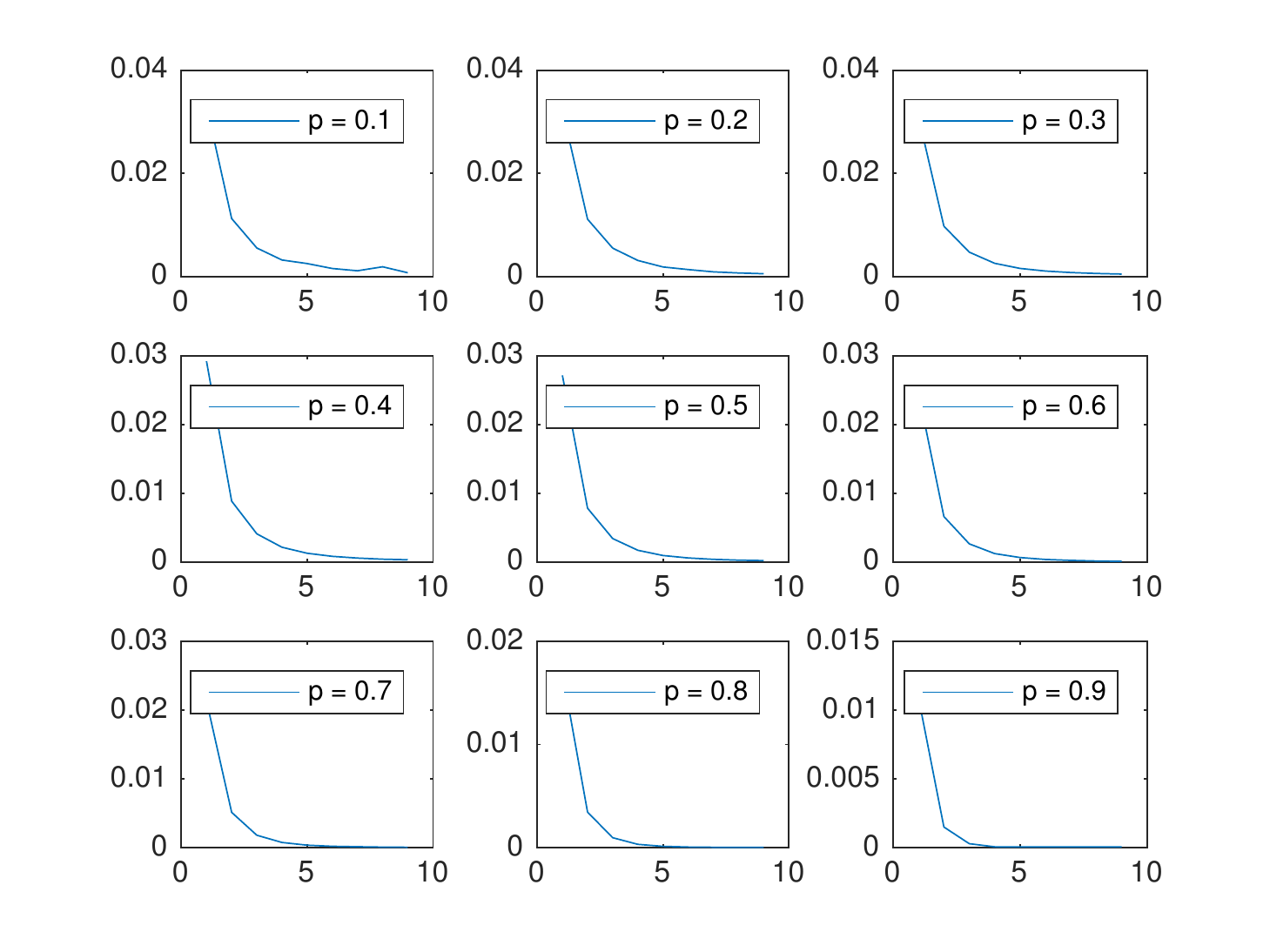}

\includegraphics[height=.32\textheight, width = .85\textwidth]{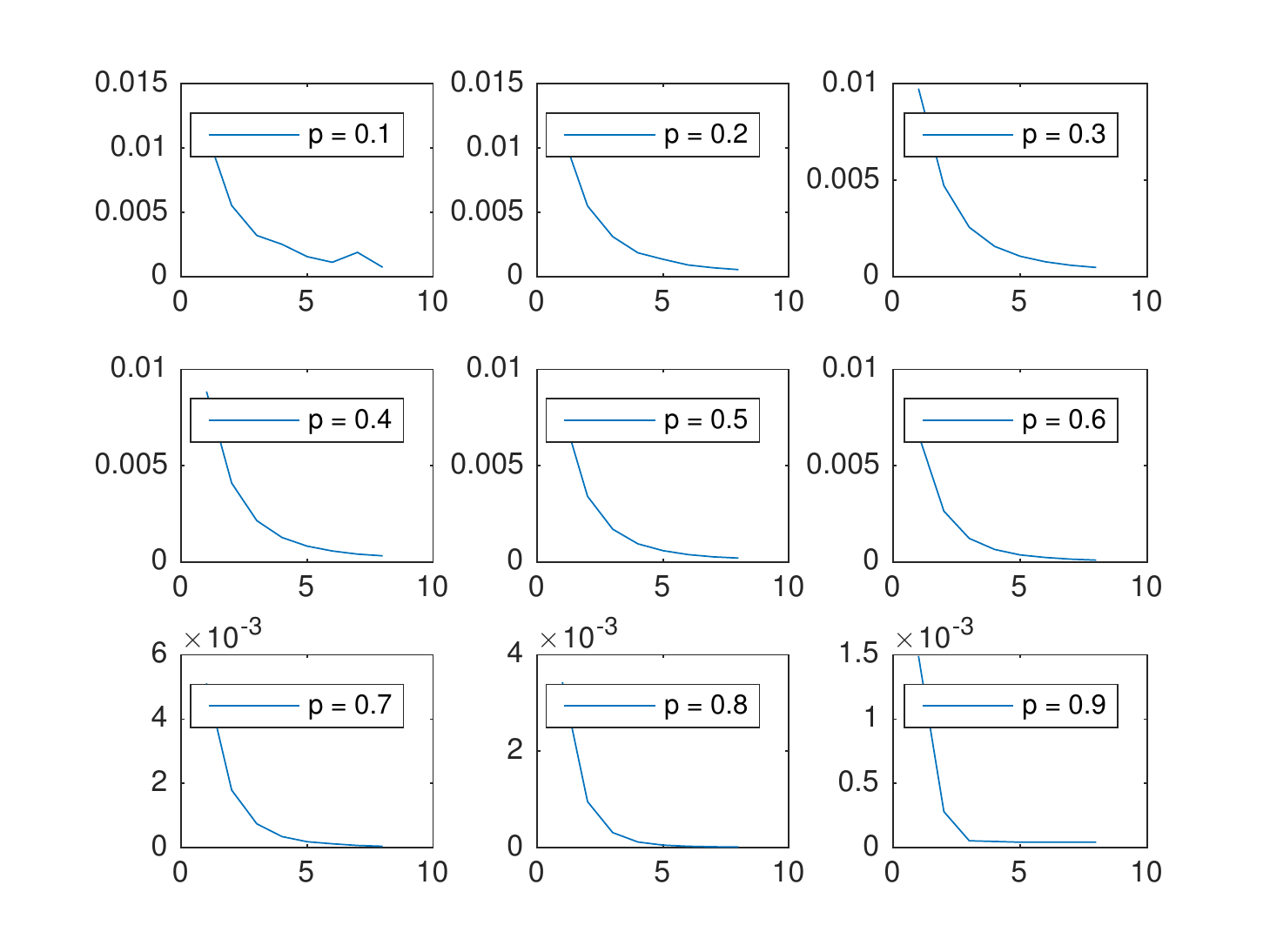}
 \caption{Relative error of the sequence of iterations. Starting from $k=1,2$, and $3$, respectively.}\label{fig:Wconv1}
\end{figure}

\subsection{$q$-controllability of shearlets}

Finally we give a numerical outcome of the $q$-controllability of the shearlet frame bounds \eqref{eq:AssumptionFrameBounds} in
Table \ref{table:FrameBounds} as this is important in practice. In fact, the numbers show that one can choose $
q \approx 20000$ in this particular case. It is important to have this number $q$ as large as possible since this in turn allows 
$s$ to be large in Theorem \ref{theorem:lpMinimizationAnalysis}.

\begin{table}[h!]
\centering
\begin{tabular}{@{}|c|c|c|c|}
\hline
Scale $J$ & $c_1/c_2$ & $N^{-1}$ & $ n \times n$\\
\hline 
2 & 0.1152& 7.266e-7& $256 \times 256$\\
\hline
3& 0.0889& 3.913e-7 & $256 \times 256$ \\
\hline
4& 0.883& 6.69e-8 & $512 \times 512$\\
\hline
5& 0.0628& 4.19e-8 & $512 \times 512$\\
\hline
6& 0.0527 & 7.6e-9 &$1024 \times 1024$\\
\hline
\end{tabular}
\caption{Verification of \eqref{eq:AssumptionFrameBounds}. In theses cases $q$ can be chosen to be $\approx$ 1e6. }\label{table:FrameBounds}
\end{table}

However, the numbers presented in Table \ref{table:FrameBounds} are very pessimistic and have to be interpreted in the right 
contest as the discussion of the next section will show.

\section{Redundancy versus sparsity in compressed sensing}\label{sec:RedundancyCS}

In the previous content of this article we have only considered the analysis formulation for redundant transforms. This is 
theoretically necessary, otherwise if the transform corresponds to a basis there is no point in distinguishing between
the analysis and synthesis formulation as these two problem would be equivalent. It has also been observed in applications
that redundant transforms can yield better results. This is for example the case for the redundant wavelet transform in image 
restoration \cite{StarMurFad}. From that point of view redundancy greatly helps and one might argue that it is also needed or 
at least desired in certain applications. The purpose of this section is to argue that although redundancy seems to yield a 
great benefit, one has to be careful and discuss: How much redundancy is good in practice? Moreover, the redundancy factor should 
be a discussion on its own and should not be confused with the results of this paper. 

We now consider two natural images of pixel size $2848\times2848$, shown in Figure \ref{fig:EagleSun} and demonstrate that typical 
(redundant) sparsifying transforms such as the wavelet and shearlet transform are from a compressed sensing point of view too 
redundant.
\begin{figure}[H]
\centering
  \includegraphics[width=.3\textwidth]{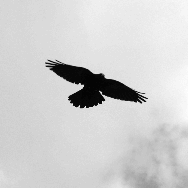}
 \qquad \qquad
\includegraphics[width=.3\textwidth]{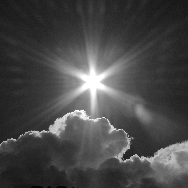}
\caption{Reference images}\label{fig:EagleSun}
\end{figure}

\subsection{Redundant wavelet transform}

We first test the redundant wavelet transform that is made available by van den Berg and Friedlander in the \texttt{spot} 
toolbox which can be downloaded at

\begin{center}
	\texttt{http://www.cs.ubc.ca/labs/scl/spot/index.html}
\end{center}

We first compute the wavelet decomposition of both reference images shown in Figure \ref{fig:EagleSun} for different total
number of scales $J = 2,4,6$. Out of these wavelet coefficients we have computed the best $s$-term approximation using 
$s = 90\%$ of the total number of pixels. This should not be confused with the number of total coefficients which is much
larger. In the notation of our previous results we let $s = 9/10 \cdot n$, where $n = 2848^2$ the dimension of the ambient
space. The reconstructions are shown in Figure \ref{fig:WaveletsWOsubsampling}.

\begin{figure}[H]
\centering
 \includegraphics[width=.3\textwidth]{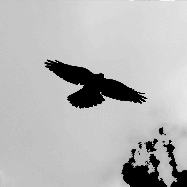}
 \ 
 \includegraphics[width=.3\textwidth]{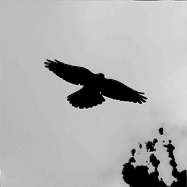}
 \ 
 \includegraphics[width=.3\textwidth]{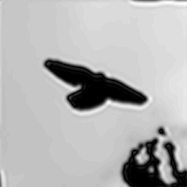}

 \includegraphics[width=.3\textwidth]{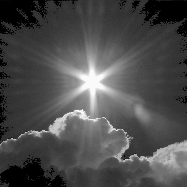}
 \ 
 \includegraphics[width=.3\textwidth]{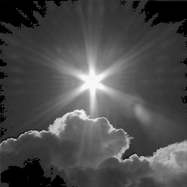}
 \ 
 \includegraphics[width=.3\textwidth]{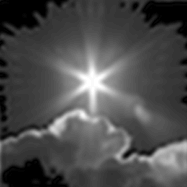}
 \caption{90\% of ambient dimension without subsampling (wavelets). Each column represents the reconstruction using
 2, 4, and 6 scales respectively.}\label{fig:WaveletsWOsubsampling}
\end{figure}
Obviously, the reconstructions get worse if more scales are available. More interestingly, the more coefficients there are 
available -- obtained by increasing the number of scales -- the more weight does the low frequency part gets in terms of the 
magnitudes of the coefficients resulting in an image that is very blurry. That shows that many more coefficients than the 
ambient dimension are needed in this case.

Now we conduct the same experiment again except that we do some subsampling before we take the best $s$-term approximation. More
precisely, we use the same coefficients but only consider every forth wavelet coefficient and set all other coefficients to zero.
Thus we divide the initial redundancy by a factor of 4. We then take the best $s$-term approximation using $s= 30\%$ of the 
total number of pixels. The outcome can be seen in Figure \ref{fig:WaveletsWsubsampling}. The reconstructed images in Figure 
\ref{fig:WaveletsWsubsampling} are significantly better than those obtained in Figure \ref{fig:WaveletsWOsubsampling}. From a 
theoretical point of view this behaviour is expected as the redundancy does not improve the approximation rate. However,
it also hints that the intrinsic sparsity of the image in the analysis coefficients $(\langle x, \psi_\lambda\rangle)_\lambda$
is much much smaller than what we can observe in the $s$-term approximation.

As we have already considered the shearlet transform in Section \ref{section:Numerics} of this paper we next want to show that 
this \emph{curse of redundancy} can also be observed in that particular case. 

\begin{figure}[H]
\centering
 \includegraphics[width=.3\textwidth]{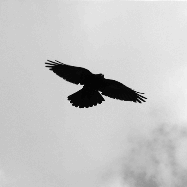}
 \ 
 \includegraphics[width=.3\textwidth]{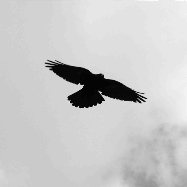}
 \ 
 \includegraphics[width=.3\textwidth]{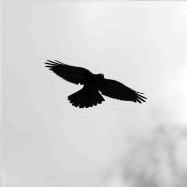}

 \includegraphics[width=.3\textwidth]{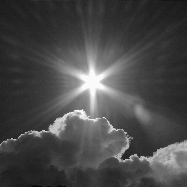}
 \ 
 \includegraphics[width=.3\textwidth]{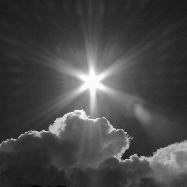}
 \ 
 \includegraphics[width=.3\textwidth]{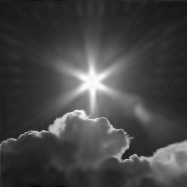}
 \caption{30\% of ambient dimension with subsampling (wavelets). Each column represents the reconstruction using
 2, 4, and 6 scales respectively.}\label{fig:WaveletsWsubsampling}
\end{figure}

\subsection{Redundant shearlet transform}

In this section we show the same experiment but with shearlets instead of wavelets. It is important to mention that the 
shearlet transform is truly redundant in the sense that there exists no non-redundant shearlet transform in the literature
so far.

\begin{figure}[H]
\centering
 \includegraphics[width=.3\textwidth]{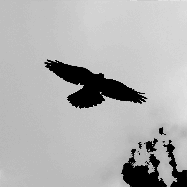}
 \ 
 \includegraphics[width=.3\textwidth]{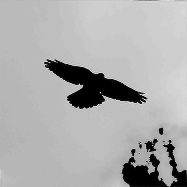}
 \ 
 \includegraphics[width=.3\textwidth]{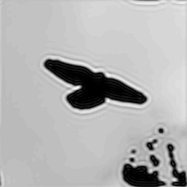}

 \includegraphics[width=.3\textwidth]{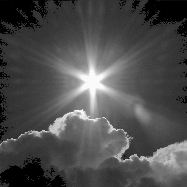}
 \ 
 \includegraphics[width=.3\textwidth]{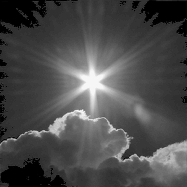}
 \ 
 \includegraphics[width=.3\textwidth]{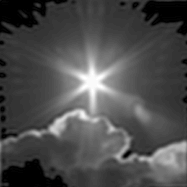}
 \caption{90\% of ambient dimension without subsampling (shearlets)  using 2, 4, and 6 scales.}\label{fig:ShearletsWOsubsampling}
\end{figure}
In Figure \ref{fig:ShearletsWOsubsampling} we show the best $s$-term approximation using again 90\% of the ambient dimension 
for a 2,4, and 6, level shearlet decomposition, respectively. We again observe that the images get more blurry
 in Figure \ref{fig:WaveletsWOsubsampling}. Similar as for wavelets we can observe that the low frequency part gets more and 
 more important if more scales are activated. 
 
Now we again subsample the already computed shearlet coefficients by only considering every forth coefficient and delete all 
others by setting them to zero. The reconstructions are shown in Figure \ref{fig:ShearletsWsubsampling}. Again, one can observe 
a significant improvement using only 10\% of the ambient dimension. Note that by increasing the scales and in that way adding 
more elements to the dictionary the system gets more and more redundant and the reconstructions get more and more accurate. This
may sound confusing as we observed the opposite behaviour in Figure \ref{fig:ShearletsWOsubsampling}. There we saw that the 
redundancy made the $s$-term approximation worse.

\begin{figure}[H]
\centering
 \includegraphics[width=.3\textwidth]{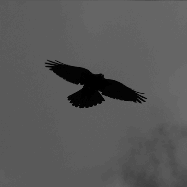}
 \ 
 \includegraphics[width=.3\textwidth]{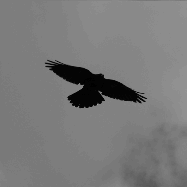}
 \ 
 \includegraphics[width=.3\textwidth]{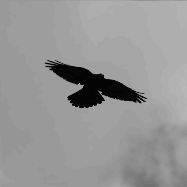}

 \includegraphics[width=.3\textwidth]{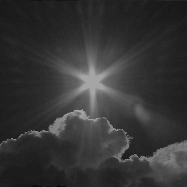}
 \ 
 \includegraphics[width=.3\textwidth]{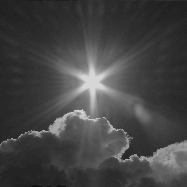}
 \ 
 \includegraphics[width=.3\textwidth]{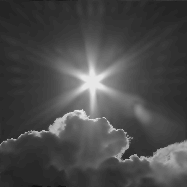}
 \caption{10\% of ambient dimension with subsampling (shearlets) using 2, 4, and 6 scales.}\label{fig:ShearletsWsubsampling}
\end{figure}

It is very evident from both numerical experiments that there is an intrinsic sparsity contained in the analysis coefficients
of a redundant transform that is not fully characterized by the sparsity assumption alone. One possible approach to tackle
this problem from a different perspective could be to rely on the \emph{statistical dimension}, \cite{LivingOnTheEdge}. This
is, however, not part of this work. Further, note that it is not correct to say, that the compression rate is not strong enough
to make the results of this paper to work as Figure \ref{fig:ShearletsWsubsampling} shows that clearly if one were to use a 
subsampled shearlet transform then the compression rate is good enough for our results to apply.

\section{Discussion and future work}

As we discussed in Section 2, it is the authors believe that the minimization problem should be performed over dual coefficients 
instead of transform coefficients if the $\Psi$-RIP is to be assumed. Surely, if the dual system and the primal system give rise
to the same sparsity pattern than this argument should no longer be valid which was in our analysis demonstrated by the concept
of identifiable duals. However, the design of sparse representation systems that have good duals is a challenging task in 
applied harmonic analysis. Nevertheless, it is necessary if one wants to combine compressed sensing with redundant dictionaries.

Furthermore, we want to comment on some problems that are left for future work:
\begin{itemize}
\item[--] In order to find the optimal $p$ the constants $C_1(p), C_2(p)$ should be optimized. For this one has to optimizing 
over $p, \gamma_1, \gamma_2, \delta, M,s $.
\item[--] The class of frames that have an identifiable dual can be seen as a generalization of scalable frames. Since the latter
have nice geometric characterizations \cite{KutOkoPhiTul} via open quadratic cones, it is interesting to see whether similar
characterizations can be computed for frames that have an identifiable dual.
\item[--] Another question left for future work is the replacement of the $\Psi$-RIP by an isometry condition as it has 
been done for the synthesis model in \cite{CanPla} in order to have a RIPless theory for the analysis formulation of the 
$\ell^p$-minimization problem. 
\item[--] Also an infinite dimensional scenario could be investigated for the $\ell^p$ scenario, in particular, in this case 
the $q$-controllability\eqref{eq:AssumptionFrameBounds} should be replaced by another condition that may shed more light on the 
redundancy problem.
\item[--] The redundancy is a very subtle issue. In Section \ref{sec:RedundancyCS} we have seen in Figure 
\ref{fig:WaveletsWsubsampling} and \ref{fig:SparsityLevels}, respectively, that there is an intrinsic sparsity structure
that has not been captured by the theory yet. It also shows that sparsity alone as in the classical sense, is not sophisticated
enough to explain why the analysis formulation works. In particular, the subsampling issue is coming from the implementations 
of such transforms, in our cases wavelets and shearlets, not from the actual theory. However, a possible point of future work
is to mathematically quantify this redundancy that arises in the discrete setting as presented in Section \ref{sec:RedundancyCS}.
A possible approach is to involve the statistical dimension or other geometric properties of the problem, 
\cite{LivingOnTheEdge,RomanGeometry}.
\end{itemize}
 
\section*{Acknowledgements}
The author would like to thank Ben Adcock for inspring discussions and acknowledges support from the Berlin Mathematical School 
as well as the DFG Collaborative Research Center TRR 109 "Discretization in Geometry and Dynamics". 

\begin{appendix}
\section{Proof of Theorem \ref{theorem:RIPKraNeeWar2}}

The optimal $\delta_s$ for which $\widetilde{A}= \begin{pmatrix} r_1 \\ \vdots \\ r_m \end{pmatrix} \in \R^{m \times n}$ 
satisfies the $\Psi$-RIP is given by
\begin{align*}
	\delta_s = \|\widetilde{A}^* \widetilde{A} - \Id_n\|_\Delta := \sup_{f \in \Delta} \langle(\widetilde{A}^* \widetilde{A} - \Id_n)x,x\rangle 
\end{align*}
where
\begin{align*}
\Delta = \{ x \in \ran \Psi^* \, : \, x = \Psi^* c, \|c\|_0 \leq s, \| x \| \leq 1\} \subseteq \R^n.
\end{align*}
Since $\E r_i^*r_i = \Id_n$ we have
\begin{align*}
	\delta_s = \| \widetilde{A}^*\widetilde{A} - \Id_n \|_\Delta = \left \| \frac{1}{m} \sum_{i=1}^m \frac{1}{m} _i^* r_i - \Id_N \right \|_\Delta = \frac{1}{m} \left \| \sum_{i=1}^m (r_i^*r_i - \E r_i^*r_i) \right \|_\Delta.
\end{align*}
For a Rademacher sequence $\eps = (\eps_i)_i$ independent of $(r_i)_i$ we have by Lemma 6.7 in \cite{Rau}
\begin{align*}
	\E \delta_s = \frac{1}{m}\E \left \| \sum_{i=1}^m (r_i^*r_i - \E r_i^*r_i) \right \|_\Delta \leq \frac{2}{m} \E \left \| \sum_{i=1}^m \eps_ir_i^*r_i \right \|_\Delta,
\end{align*}
hence,
\begin{align*}
	\E \delta_s \leq \frac{2}{m} \E_r \E_\eps \sup_{x \in \Delta} \left| \langle \sum_{i=1}^m \eps_ir_i^*r_i x, x \rangle \right| 
	= \frac{2}{m} \E_r \E_\eps \sup_{x \in \Delta} \left|  \sum_{i=1}^m \eps_i |\langle r_i, x \rangle|^2 \right|.
\end{align*}
Define the pseudo-metric
\begin{align*}
	d(x,y) = \left(\sum_{i=1}^m \left(|\langle r_i, x \rangle|^2 - |\langle r_i, y \rangle|^2 \right)^2\right)^{1/2}.
\end{align*}
Then, as shown in \cite{KraNeeWar} we have for $x,y \in \Delta$
\begin{align*}
	d(x,y) \leq 2 \sup_{z \in \Delta} \left( \sum_{i=1}^m | \langle r_i ,z \rangle|^{2p}\right)^{1/(2p)} \left( \sum_{i=1}^m | \langle r_i , x-y \rangle|^{2p}\right)^{1/(2p)} ,
\end{align*}
where $p,q \geq 1$ such that $p^{-1} + q^{-1} =1$.

Now, for any $h \in \Delta$ of the form $z = \Psi^*c$ with $\|c \|_0 \leq s$ and any realization of $(r_i)_i$ we have
\begin{align*}
	| \langle r_i, z \rangle | &= | \langle \Psi(\Psi^*\Psi)^{-1} r_i, \Psi \Psi^* c \rangle | \\
	&\leq \sum_{\lambda \leq N} | \langle r_i, \psitilde_\lambda \rangle | |(\Psi \Psi^* c)_\lambda | \\
	&\leq \sum_{\lambda \leq N} \frac{K}{c_1} |(\Psi \Psi^* c)_\lambda |,
\end{align*}
where $K>0$ is so that $\| \psi_\lambda \| \leq K$ for all $\lambda \leq N$ and $c_1$ denotes the lower frame bound.
Therefore
\begin{align*}
	|\langle r_i, h \rangle | \leq \frac{K}{c_1}  L \sqrt{s}
\end{align*}
with 
\begin{align*}
L = \sup_{ \substack{\| \Psi^* c \| =1 \\ \| c \|_0 \leq s}} \frac{\|  (\Psi \Psi^* c)_\lambda \|_1}{\sqrt{s}}.
\end{align*}
Therefore we obtain 
\begin{align*}
	\sup_{z \in \Delta} \left( \sum_{i=1}^m | \langle r_i, z \rangle|^{2p}\right)^{1/(2p)} &= \sup_{z \in \Delta} \left( \sum_{i=1}^m | \langle r_i, z \rangle|^{2}| \langle r_i, z \rangle|^{2p-2}\right)^{1/(2p)}\\
	&\leq \left(s\left(\frac{KL}{c_1}\right)^2 \right)^{(p-1)/(2p)} \left( \sup_{z \in \Delta} \sum_{i=1}^m | \langle r_i, z \rangle|^{2p}\right)^{1/(2p)}.
\end{align*}
The rest of the proof follows the argumentation given in \cite{KraNeeWar}.

For a set $\Sigma$, a metric $d$ and a given $t>0$ the \emph{covering number} $\mathcal{N}(\Sigma, d, t)$ is 
defined as the smallest number of balls of radius $t$ centered at points of $\Sigma$ necessary to cover $\Sigma$ 
with respect to $d$. By Dudley's inequality we have
\begin{align}
	\E_\eps \sup_{x \in \Delta} \left| \langle \sum_{i=1}^m \eps_i r_i^*r_i x, x \rangle \right| \leq 4 \sqrt{2} \int_0^\infty \sqrt{ \log( \mathcal{N}(\Delta, d, t)} \, dt. \label{eq:ExpCovNum}
\end{align}
Using the semi-norm
\begin{align*}
	\| x \|_{X,q} := \left( \sum_{i=1}^m | \langle r_i, x \rangle|^{2q}\right)^{1/(2q)}
\end{align*}
we obtain using covering arguments and \eqref{eq:ExpCovNum}
\begin{align*}
	\E_\eps \sup_{x \in \Delta} &\left| \langle \sum_{i=1}^m \eps_i r_i^*r_i x, x \rangle \right| \\
	&\leq C \left(s\left(\frac{KL}{c_1}\right)^2 \right)^{(p-1)/(2p)}  \left( \sum_{i=1}^m | \langle r_i, x \rangle|^{2q}\right)^{1/(2q)} \int_0^\infty \sqrt{ \log( \mathcal{N}(\Delta, \| \cdot \|_{X,q}, t)} \, dt.
\end{align*}
Now, following the arguments in \cite{KraNeeWar}  we have
\begin{align*}
	\int_0^\infty \sqrt{ \log( \mathcal{N}(\Delta, \| \cdot \|_{X,q}, t)} \, dt \leq C \sqrt{q(sL^2)m^{1/q}\log(n) \log^2(sL^2)}.
\end{align*}
Thus in \eqref{eq:ExpCovNum} we obtain
\begin{align*}
	\E \delta_s &\leq \frac{C\left(s\left(\frac{KL}{c_1}\right)^2 \right)^{(p-1)/(2p)}  \sqrt{qm^{1/q}sL^2\log(n)\log^2(sL^2)}}{m} \E \sup_{x \in \Delta}\left(\sum_{i=1}^m | \langle r_i, x \rangle|^2 \right)^{1/(2p)} \\
	\leq& \frac{C\left(s\left(\frac{KL}{c_1}\right)^2 \right)^{(p-1)/(2p)} \sqrt{q\log(n)\log^2(sL^2)}}{m^{1-1/(2q)-1/(2p)}} \E\left( \frac{1}{m} \| \sum_{i=1}^m r_i^*r_i - \Id_n \|_\Delta + \| \Id_n \|_\Delta\right)^{1/(2p)} \\
	\leq& \frac{C\left(s\left(\frac{KL}{c_1}\right)^2 \right)^{(p-1)/(2p)} \sqrt{q\log(n)\log^2(sL^2)}}{m^{1/2}} \sqrt{\E\delta_s + 1}.
\end{align*}
We can assume $K/c_1$ to greater than one, hence
\begin{align*}
	\E \delta_s
	\leq \frac{C\left(s\left(\frac{KL}{c_1}\right)^2 \right)^{(p-1)/(2p)} \sqrt{q\log(n)\log^2(s\left(\frac{KL}{c_1}\right)^2)}}{m^{1/2}} \sqrt{\E\delta_s + 1}.
\end{align*}
Choosing $p = 1+ (\log(s(KL)^2c_1^{-2}))^{-1}$ and $q = 1+ \log(s(KL)^2c_1^{-2})$ yields
\begin{align*}
	(s(KL)^2c_1^{-2})^{1/2+(p-1)/(2p)} \leq \sqrt{e},
\end{align*}
hence,
\begin{align*}
	\E\delta_s \leq C \sqrt{2\log(n) \log^2(s(KL)^2c_1^{-2})/m} \sqrt{\E \delta_s + 1}.
\end{align*}
Finally,
\begin{align*}
	\E \delta_s \leq C \sqrt{\frac{s(KL)^2c_1^{-2} \log(n) \log^3(sL^2)}{m}}
\end{align*}
provided $\frac{s(KL)^2c_1^{-2} \log(n) \log^3(s(KL)^2c_1^{-2})}{m}\leq 1$. Therefore, $\E \delta_s \leq \delta/2$ for some $\delta \in (0,1)$ if
\begin{align}
	m \geq C \delta^{-2}s(KL)^2c_1^{-2} \log^3(s(KL)^2c_1^{-2}) \log N. \label{eq:BoundForM}
\end{align}
Let $f_{x,y}(r) = \Rea ( \langle (r_i^*r_i - \Id_n)z,w \rangle)$ so that
\begin{align*}
 m \delta_s = \left \|\sum_{i=1}^m \left( r_i^* r_i - \E r_i^*r_i \right) \right\|_\Delta 
 = \sup_{x,y \in \Delta} \sum_{i = 1}^M f_{x,y}(r_i).
\end{align*}
Note that we have
\begin{itemize}
 \item[$\triangleright$] $\E f_{x,y}(r_i) = 0$,
 \item[$\triangleright$] $|f_{x,y}(r)| \leq  s (KL)^2c_1^{-2} +1$,
 \item[$\triangleright$] $\E|f_{x,y}(r)|^2 = \E\|(r_i^*r_i - \Id)x \|_2^2 \leq (s (KL)^2c_1^{-2} +1)^2$.
\end{itemize}
Now, fix some $\delta \in (0,1)$ and choose $m$ in accordance with \eqref{eq:BoundForM}. Then by Theorem 6.25 of \cite{Rau}
we haven
\begin{align}
 \mathbb{P}(\delta_s \geq \delta) &\leq \mathbb{P}(\delta_s \geq \E \delta_s + \delta/9) \nonumber\\
 &= \mathbb{P} \left( \left\| \sum_{i=1}^m(r_i^* r_i - \E r_i^*r_i)\right\|_\Delta \geq \E \left\| \sum_{i=1}^m(r_i^* r_i - \E r_i^*r_i)\right\|_\Delta + \delta m/9\right)\nonumber \\
&\leq \exp\left( - \frac{\left(\frac{\delta m}{9(s(KL)^2c_1^{-2} + 1)}\right)^2}{2m\left(1+ \frac{\delta}{s(KL)^2c_1^{-2} +1}\right) + \frac{2}{3} \left( \frac{\delta m}{9(s\eta^2 + 1}\right) }\right)\nonumber\\
&\leq \exp\left( - \frac{\delta^2 m}{Cs(KL)^2c_1^{-2}} \right), \label{eq:FinalBoundForM}
 \end{align}
 where the constant $C$ might changed in the last estimate. Further, if 
 \begin{align*}
 m \geq C \delta^{-2} s(KL)^2c_1^{-2}\log(1/\gamma),
 \end{align*} 
 then
 \eqref{eq:FinalBoundForM} is bounded by $\gamma$. Thus, $\delta_s \leq \delta$ with probability $1-\gamma$ if
 \begin{align*}
  m \geq C \delta^{-2}s (KL)^2c_1^{-2} \max\{ \log^3(s (KL)^2c_1^{-2}) \log(N), \log(1/\gamma)\}.
 \end{align*}

The proof is complete.

\end{appendix}
\bibliographystyle{abbrv}
\bibliography{bib}

\end{document}